\documentclass[letterpaper,11pt,pra,aps,superscriptaddress,showpacs]{revtex4-1}

\usepackage{amsmath}
\usepackage{amssymb}
\usepackage{amsthm}
\usepackage{amsfonts}
\usepackage[all]{xy}
\usepackage{graphicx}
\usepackage{hyperref}

\newtheorem{theorem}{Theorem}
\newtheorem{lem}{Lemma}

\newtheorem{prop}{Proposition}
\theoremstyle{definition} \newtheorem{defn}{Definition}

\newcommand{\ket}[1]{\left| #1 \right\rangle}
\newcommand{\bra}[1]{\left\langle #1 \right|}

\newcommand{\half}{\frac{1}{2}}
\newcommand{\nn}{\nonumber}
\newcommand{\ltqo}{{L_{tqo}}}
\newcommand{\calD}{{\cal D }}
\newcommand{\calL}{{\cal L }}
\newcommand{\calZ}{{\cal Z }}
\newcommand{\ZZ}{\mathbb{Z}}
\newcommand{\trace}{\mathop{\mathrm{Tr}}\nolimits}

\begin{document}

\title{Analytic and numerical demonstration of quantum self-correction in the 3D Cubic Code}

\author{Sergey \surname{Bravyi}}
\affiliation{IBM Watson Research Center,  Yorktown Heights,  NY 10598}
\author{Jeongwan \surname{Haah}}
\affiliation{Institute for Quantum Information and Matter, California Institute of Technology, Pasadena, CA 91125}

\date{14 December 2011}

\begin{abstract}
A big open question in the quantum information theory concerns feasibility of a self-correcting quantum memory.
A quantum state recorded in such memory can be stored reliably for a macroscopic time without need for
active error correction if the memory is put in contact with a cold enough thermal bath.
In this paper we derive a rigorous lower bound on the memory time $T_{mem}$ of the 3D Cubic Code model which was recently
conjectured to have a self-correcting behavior.
Assuming that dynamics of the memory system can be described by a Markovian
master equation of Davies form, we prove that   $T_{mem}\ge L^{c\beta}$ for some constant $c>0$, where $L$ is the lattice size
and $\beta$ is the inverse temperature of the bath.  However,
this bound applies only if the lattice size does not exceed certain critical value $L^*\sim e^{\beta/3}$.
 We also report a numerical Monte Carlo simulation of the studied memory indicating that our analytic bounds on $T_{mem}$ are tight up to constant coefficients.  In order to model the readout step we introduce a new decoding algorithm which might
 be of independent interest. Our decoder can be implemented efficiently for any topological stabilizer code and has a constant error threshold under
random uncorrelated errors.
 \end{abstract}

\pacs{03.67.Pp, 03.67.Ac, 03.65.Ud}

\maketitle

\tableofcontents

\section{Introduction}

Any practical memory device must function reliably in a presence of small
hardware imperfections and protect the recorded data against  thermal noise.
Building a memory capable of  storing quantum information is particularly challenging since
a quantum state must be protected against both bit-flip and phase-flip errors.
Furthermore, in contrast to classical bits, quantum states form a continuous set
thus being more vulnerable to small  hardware imperfections.

Ground states of topologically ordered many-body systems, such as
fractional quantum Hall liquids~\cite{NayakEtAl2008tqc,HormoziBenesteelSimon2009tqc} and unpaired Majorana  fermions
in nanowires and 2D heterostructures~\cite{FuKane2008proximity,
SauLutchynEtAl2010semiconductor,
OregRefaelOppen2010helical,
AliceaOregRefaelEtAl2011wire,
StanescuLutchynSarma2011semiconductor}
were proposed as a natural quantum data repository insensitive to local imperfections.
A qubit encoded into the ground subspace of a topologically ordered system is almost perfectly
decoupled from any local perturbation due to local indistinguishability of the ground
states~\cite{Kitaev2003Fault-tolerant,
BravyiHastingsMichalakis2010stability,
BravyiHastings2011short,
MichalakisPytel2011stability}.

To undo the effect of noise, a user of any memory,  either classical or quantum, must invoke some form of
error correction.   It was shown by Dennis et al~\cite{DennisKitaevLandahlEtAl2002Topological} that topological memory
based on the Kitaev's 2D toric code model~\cite{Kitaev2003Fault-tolerant} can
tolerate stochastic local noise provided that error correction is performed frequently enough to prevent
errors from accumulating. However, such active error correction would require an extensive and fast classical input/output to
the quantum hardware which might pose a challenge for building large-scale memory devices.

An intriguing open question raised in~\cite{DennisKitaevLandahlEtAl2002Topological,Bacon2006Operator} is whether topological memories can be
{\em self-correcting}, that is, whether active error correction can be
imitated by the natural dynamics of the memory system coupled to a thermal bath.
The physical mechanism behind self-correction envisioned in~\cite{Bacon2006Operator}
relies on a presence of  ``energy barriers" separating distinct ground states
and energy dissipation.
If the energy barriers are high enough
and the evolution time is not too long,
the memory system will be  locked in the  ``energy valley" surrounding the initial ground state.
A user of a  self-correcting memory would
only be responsible for preparation of the initial ground state and performing one final round of active error correction
at the readout step. The latter is required to clean up the residual low-energy errors.
In contrast to the active error correction approach of~\cite{DennisKitaevLandahlEtAl2002Topological}, the storage itself would require no action
from the user whatsoever.

The nature of excitations in a topological memory plays a crucial role in assessing its self-correcting
capability. Anyons in the 2D toric code~\cite{Kitaev2003Fault-tolerant} provide a paradigmatic example of such excitations.
A fundamental  flaw of 2D topological memories that rules out  self-correction
is the lack of energy barriers that could suppress diffusion of anyons over large distances~\cite{DennisKitaevLandahlEtAl2002Topological}.
For instance, consider 
a process that involves a creation of anyon pair from the ground state,
a transport of one anyon along a non-contractible loop on the torus,
and a final annihilation of the pair. This process 
enacts a non-trivial transformation on the ground subspace. However, 
it can be  implemented by a stream of local errors at a constant energy cost.
The lack of self-correction for the 2D toric code model  was rigorously confirmed by Alicki et al~\cite{AlickiFannesHorodecki2008thermalization} who showed that the relaxation time towards the equilibrium state is a constant independent of the lattice size.
A more general no-go result for quantum self-correction based on arbitrary 2D stabilizer code Hamiltonians
was derived in~\cite{BravyiTerhal2008no-go, KayColbeck2008Quantum}.

While anyons have a rich algebraic structure that depends on a particular model, their common property
is that anyons are {\em topological defects}.
More broadly, a topological defect is  a point-like excitation that cannot be created from the ground state by a local operator  without creating other excitations.
A domain wall in the 1D ferromagnetic Ising chain provides
the simplest example of a topological defect.
The second property that applies to all models with anyonic excitations is that
anyons are {\em mobile} topological defects. In other words,
once an anyon has been created at some spatial location, a local operator can move it to another location
nearby without creating any other excitations. As we argued above, the latter property is the key reason why self-correction is hard, if not impossible, to achieve in 2D.
It raises a question whether the laws of physics permit
topological phases of matter without mobile topological defects~?
(In the coding language it translates to existence of topological codes without string-like logical operators.)
Quite surprisingly, the answer turns out to be yes. A family of 3D topological memory models
without mobile topological defects was recently discovered by one of the authors~\cite{Haah2011Local}.
One of these models, which we call the 3D Cubic Code, obeys the so called no-strings rule~\cite{BravyiHaah2011Energy}.
It establishes a constant upper bound $\alpha$ on the mobility range of topological defects.
More precisely, given any isolated topological defect occupying a spatial region of size $R$, no local operator can move this defect distance  $\alpha R$ or more away
without creating extra excitations, see~\cite{Haah2011Local} for details.
A strong evidence in favor of self-correcting properties of the 3D Cubic Code has been recently obtained
in~\cite{BravyiHaah2011Energy} by  proving that the energy
barrier separating distinct ground states grows as a logarithm of the lattice size.

The main goal of the present paper is to assess self-correcting properties of the 3D Cubic Code
in a more direct way by  calculating its memory time as a function of the lattice size and the temperature.
For the sake of concreteness, we shall state our results only for the 3D Cubic Code, but our analysis applies
to any quantum memory model based on a stabilizer code with local generators
which obeys the topological order condition and the no-strings rule as stated in~\cite{BravyiHaah2011Energy}.

\section{Memory Hamiltonian, thermal noise, and decoding}
\label{sec:misc}

In order to use any memory, either classical or quantum, a user must be able to write, store,
and read information. In this section we describe these steps formally for a topological quantum memory
based on the 3D Cubic Code. Definition of the code Hamiltonian and its basic
properties are summarized in Section~\ref{subs:CubicCode}.
Our formal model of the thermal noise based on Davies weak-coupling limit is presented in Section~\ref{subs:Davies}.
The readout step consists of a syndrome measurement and an error correction, see Section~\ref{subs:decoding}.

\subsection{3D Cubic Code Hamiltonian}
\label{subs:CubicCode}

\begin{figure}
\newcommand{\drawgenerator}[8]{%
\xymatrix@!0{%
& #8 \ar@{-}[ld]\ar@{.}[dd] \ar@{-}[rr] & & #7 \ar@{-}[ld]  \\%
#1 \ar@{-}[rr] \ar@{-}[dd] &  & #2 \ar@{-}[dd] &            \\%
& #6 \ar@{.}[ld] &  & #5 \ar@{-}[uu] \ar@{.}[ll]       \\%
#3 \ar@{-}[rr] &  & #4 \ar@{-}[ru]                       %
}%
}
\centering
\begin{tabular}{ccc}
$ \drawgenerator{ZI}{ZZ}{IZ}{ZI}{IZ}{II}{ZI}{IZ} $
& $\quad$ &
$ \drawgenerator{XI}{II}{IX}{XI}{IX}{XX}{XI}{IX} $ \\
$G^Z$ & $\quad$ & $G^X$
\end{tabular}
\caption{Stabilizer generators of the 3D Cubic Code. Here
$X\equiv \sigma^x$ and $Z\equiv \sigma^x$ represent single-qubit Pauli operators,
while $I$ is the identity operator.
Double-letter indices represent two-qubit Pauli operators,
for example,
$IZ\equiv I\otimes Z$, $ZZ\equiv Z\otimes Z$, $II\equiv I\otimes I$ etc.}
\label{fig:CubicCode}
\end{figure}
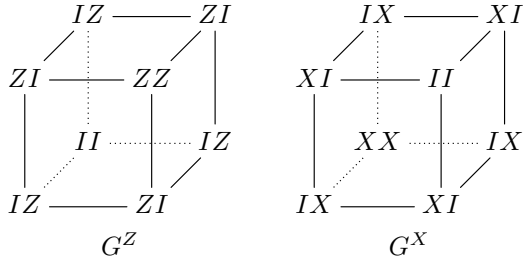

Let $\Lambda=\ZZ_L\times \ZZ_L \times \ZZ_L$ be the regular 3D Cubic lattice
of linear size $L$ with periodic boundary conditions. Each site of the lattice is occupied by two qubits.
The 3D Cubic Code Hamiltonian introduced in~\cite{Haah2011Local} has the following form:
\begin{equation*}
H=-J \sum_{c} G^X_c + G^Z_c,
\end{equation*}
where the sum runs over all $L^3$ elementary cubes $c$ and the operators $G^X_c$, $G^Z_c$
act on the qubits of $c$ as shown on Fig.~\ref{fig:CubicCode}.
The positive coupling constant $J$ will be set to $J=\half$ for simplicity.
We shall refer to operators $G^X_c$, $G^Z_c$ as {\em stabilizer generators}, or simply
stabilizers. Note that each stabilizer acts non-trivially only on $8$ qubits.
Let us recall some basic properties of the 3D Cubic Code, see~\cite{Haah2011Local} for details.
First, one can easily check that the stabilizers $G^X_c$, $G^Z_{c'}$ commute with each other for all $c,c'$.
A ground state of $H$ is a common $+1$ eigenstate of all stabilizers.
The degeneracy of the ground states is $2^{k(L)}$ for some integer $2\le k(L)\le 4L$.
The ground subspace of $H$ has topological order, that is, different ground states
cannot be distinguished locally. More precisely, if $O$ is any operator
whose support can be bounded by a cubic box of size $<L$ then the restriction of $O$
onto the ground subspace is proportional to the identity operator.
Excited states of $H$ can be described by configurations of defects, that is, stabilizers
whose eigenvalue is $-1$. Each defect costs one unit of energy.

\subsection{Thermal noise}
\label{subs:Davies}

Suppose at time $t=0$ the memory system is initialized in some ground state $\rho(0)$
encoding a quantum state to be stored.
We shall model interaction between the memory system and the thermal bath
using the Davies weak coupling limit~\cite{Davies1974}. It provides a Markovian
master equation of the following form:
\begin{equation}
\dot{\rho}(t)=-i[H,\rho(t)] + \calL(\rho(t)), \quad t\ge 0.
\label{eq:Markovian-master}
\end{equation}
Here $\rho(t)$ is the state of the memory system at time $t$ and $\calL$ is the Lindblad generator describing dissipation of energy. To define $\calL$, let us choose some set of self-adjoint operators $\{A_\alpha\}$ through which
the memory can couple to the bath.  We shall assume that each $A_\alpha$
acts non-trivially on a constant number of qubits.
For example, $\{A_\alpha\}$ could be the set
of all single-qubit Pauli operators. Let $A_\alpha=\sum_{\omega} A_{\alpha,\omega}$,
where $A_{\alpha,\omega}$ is the spectral component of $A_\alpha$
that maps eigenvectors of $H$ with energy $E$ to eigenvectors with energy $E-\omega$. Then
\begin{equation}
\calL( \rho)=\sum_\alpha \sum_\omega h(\alpha,\omega) \left(
A_{\alpha,\omega} \rho A_{\alpha,\omega}^\dag - \frac12 \{  \rho,A_{\alpha,\omega}^\dag A_{\alpha,\omega} \} \right).
\label{eq:lindbladian}
\end{equation}
The coefficient $h(\alpha,\omega)$ is the rate of quantum jumps caused by $A_\alpha$ transferring energy $\omega$
from the memory to the bath. It must obey the  detailed balance condition
\begin{equation}
\label{DB}
h(\alpha,-\omega)=e^{-\beta \omega} h(\alpha,\omega),
\end{equation}
where $\beta$ is the inverse bath temperature. The detailed balance condition Eq.\eqref{DB} is the only part
of our model that depends on  the bath temperature. It guarantees that
the Gibbs state $\rho_\beta \sim e^{-\beta H}$ is the fixed point of the
dynamics, $\calL(\rho_\beta)=0$. This is a unique fixed point under certain natural ergodicity
conditions~\cite{Spohn1977}.
Furthermore, we shall assume that $\|A_\alpha\|\le 1$ and
\begin{equation}
\label{eq:weak}
\max_{\alpha,\omega} h(\alpha,\omega) = O(1).
\end{equation}
Let us remark that the Davies weak coupling limit was adopted
as a  model of the thermal dynamics
in most of the previous works with a rigorous analysis of quantum self-correction;
see for instance~\cite{AlickiHorodeckiHorodeckiEtAl2008thermal,
AlickiFannesHorodecki2008thermalization,
ChesiLossBravyiEtAl2009Thermodynamic,
ChesiRoethlisbergerLoss2009Self-Correcting}.

\subsection{Decoding}
\label{subs:decoding}

The final state $\rho(t)$ generated by the Davies dynamics can be regarded as a corrupted version of the
initial encoded state $\rho(0)$.   A decoder retrieves the encoded information from $\rho(t)$
by performing a syndrome measurement and an error correction.
A syndrome measurement involves a non-destructive eigenvalue measurement of all
stabilizer generators $G^X_c$, $G^Z_c$.
The measured syndrome $S$ can be regarded as a classical bit string that assigns an eigenvalue
$\pm 1$ to each generator. Let $\Pi_S$ be the
projector onto the subspace with syndrome $S$.

The error correction step is specified by an algorithm that takes as input  the measured syndrome
$S$ and returns a correcting Pauli operator $P_{ec}(S)$.  A good choice of the error correction algorithm
is a highly non-trivial task which we discuss in Section~\ref{sec:rgdecoder}.
The net action of the decoder on states can be described
by  a trace preserving completely positive (TPCP) linear map $\Phi_{ec}$ such that
\begin{equation}
\label{decoder}
\Phi_{ec}(\rho)=\sum_S P_{ec}(S)\Pi_S \, \rho \Pi_S P_{ec}(S)^\dag,
\end{equation}
where the sum runs over all possible syndromes.

\section{Main results and sketch of the proof}
\label{sec:results}

\begin{figure*}[p]
 \includegraphics[width=.71\textwidth]{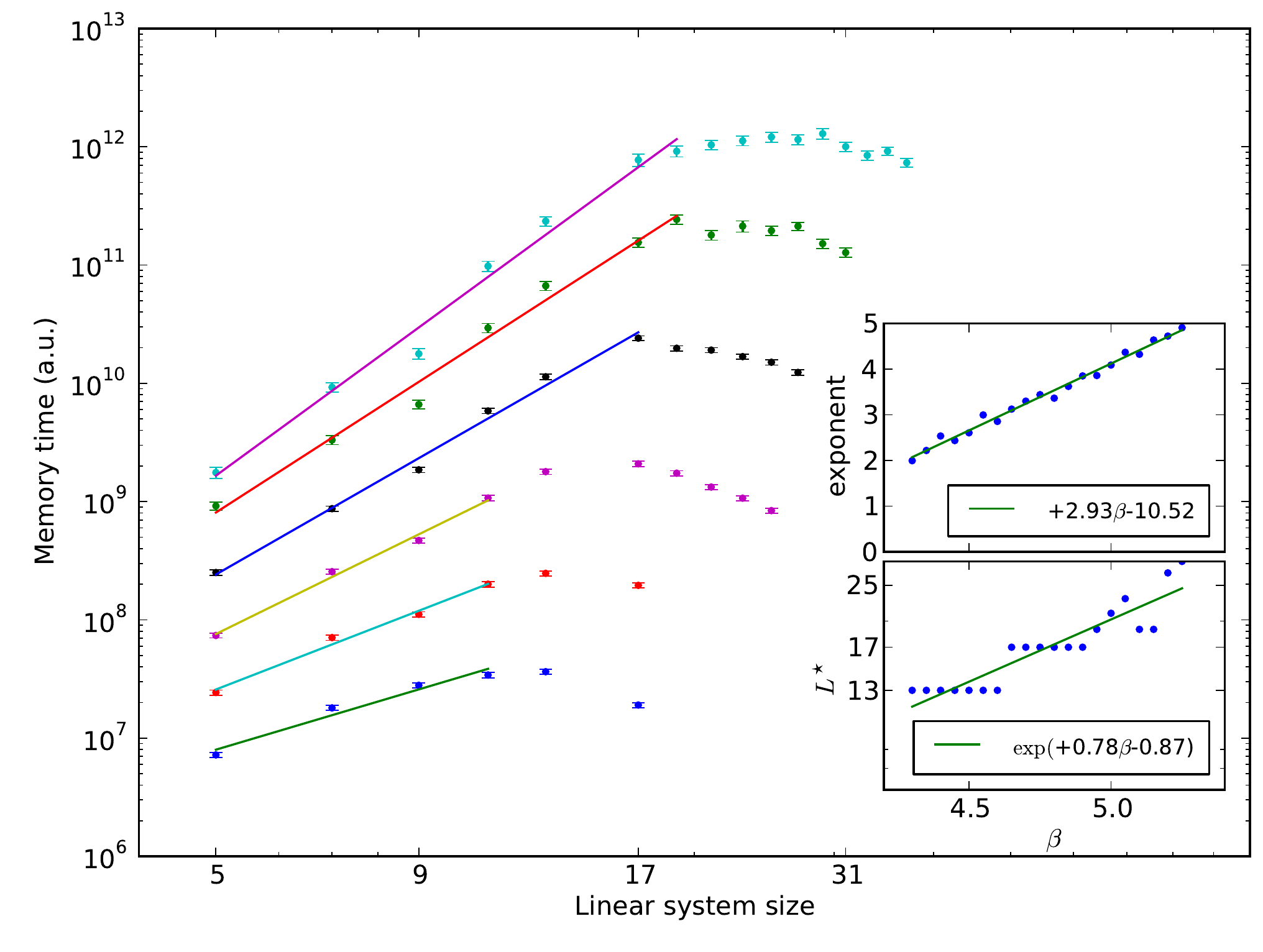}
\caption{The memory time $T_{mem}$ vs. the system size $L$. In the upper inset is shown the exponent of the power law fit of $T_{mem}$ for the first a few system sizes. It is clear that $T_{mem} \propto L^{2.93 \beta -10.5}$ when $L < L^\star$, where $L^\star$ is the optimal system size where $T_{mem}$ reaches maximum. The data for $\beta = 4.3, 4.5, 4.7, 4.9, 5.1, 5.25$ are shown.}
\label{fig:TvsL}
\end{figure*}

\begin{figure*}[p]
 \includegraphics[width=.71\textwidth]{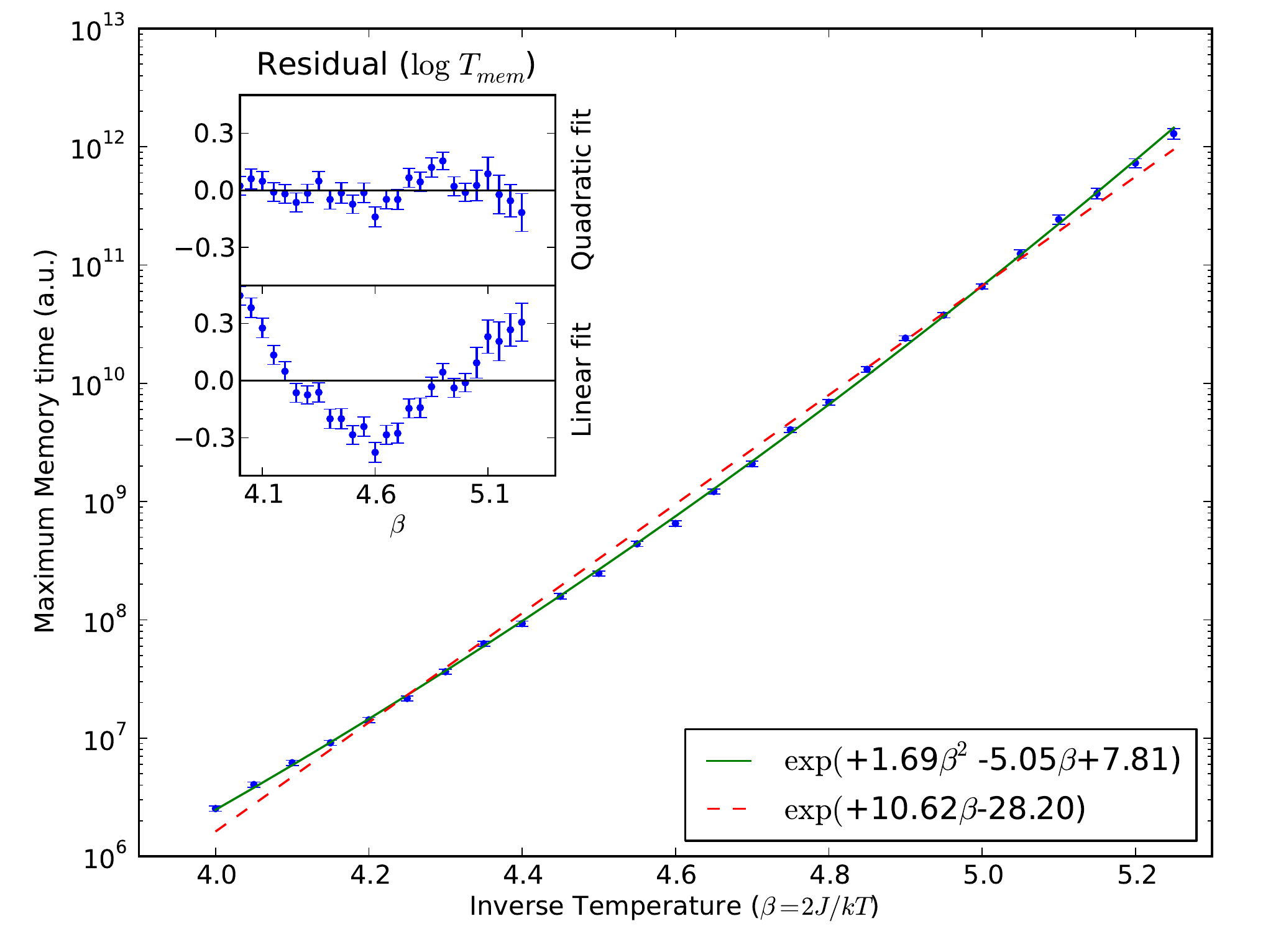}
\caption{The maximum memory time $T_{mem}$ vs. the inverse temperature $\beta$. The memory time is maximized with respect to the system size. The logarithm of $T_{mem}$ clearly follows a quadratic relation with $\beta$ as oppose to a linear one.}
\label{fig:TvsBeta}
\end{figure*}

Our first result is an upper bound on the storage error --- the trace
distance between the initial encoded state and the final error corrected state.
\begin{theorem}
\label{thm:main1}
There exists a decoder $\Phi_{ec}$ and a constant $c>0$
such that for for any inverse temperature $\beta>0$, any state $\rho(0)$ supported on the ground subspace
of $H$, any evolution time $t\ge 0$, and any constant $0<a<1$ one has
\begin{equation}
\label{error-bound}
\epsilon(t)\equiv \| \rho(0) - \Phi_{ec}(\rho(t)) \|_1 \le O(t) \cdot 2^{k(L)} \cdot L^{3-ac\beta}
\end{equation}
as long as $L\le e^{(1-a)\beta/3}$. The error correction algorithm used by the decoder
has running time $poly(L)$.
\end{theorem}
Clearly, our bound is most useful when  $k(L)$ is small.
In the following we shall mostly be interested in the
smallest ground state degeneracy, $k(L)=2$. This happens for any odd $3\le L\le 200$ such that
$L$ is not a multiple of $15$ or $63$, see~\cite{Haah2011Local}.
One can also show~\cite{Haah11+} that there exists an infinite sequence
of lattice sizes  such that $k(L)=2$, for example, $k(2^p+1)=2$ for all $p\ge 1$.
Unfortunately, a general explicit formula for $k(L)$ is unknown.

The upper bound on the storage error can be easily translated to a lower bound on the memory time.
Indeed, if one is willing to tolerate a fixed storage error $\epsilon$, say $\epsilon=0.01$,
the memory time $T_{mem}$ can be defined as the smallest $t\ge 0$ such that $\epsilon(t)\ge \epsilon$.
Assuming that the lattice size is chosen such that $k(L)=2$, Theorem~\ref{thm:main1} implies that
\begin{equation}
\label{Tmem-bound}
T_{mem}\ge L^{ac\beta-3} \quad \mbox{for any $L\le e^{(1-a)\beta/3}$}.
\end{equation}
Here we neglected the overall constant coefficient.
It shows that for low temperatures, $\beta \gg 1$,  and sufficiently small system size,
$L\ll e^{\beta/3}$,  the memory time grows polynomially with $L$,
while the degree of the polynomial is proportional to $\beta$.  To the best of our knowledge, this provides
the  first realistic example of a topological memory with a self-correcting behavior. Unfortunately, the
studied model has no self-correction in the thermodynamic limit $\beta=O(1)$ and $L\to \infty$.
In this sense, it is only partially self-correcting.
At the optimally chosen lattice size,
the maximum memory time $T_{mem}(\beta)$ achievable at a given temperature $\beta$
is easily found from Eq.~(\ref{Tmem-bound}):
\begin{equation}
\label{Tmem-bound1}
T_{mem}(\beta) \ge e^{c\beta^2/12}
\end{equation}
for $\beta\gg 1$. For comparison, the memory time of the 2D toric code model grows
only exponentially with $\beta$, see~\cite{AlickiFannesHorodecki2008thermalization, ChesiRoethlisbergerLoss2009Self-Correcting}.
Depending on the value of the constant $c$ and the temperatures realizable in experiments,
the scaling Eq.~(\ref{Tmem-bound1}) may be favorable enough to achieve macroscopic memory times.

The restriction $L \ll e^{\beta/3}$ implies that the average number of defects (flipped stabilizers)
in the equilibrium Gibbs state $\rho_\beta \sim e^{-\beta H}$
is less than one, that is, the Gibbs state has most of its weight on
the ground subspace of $H$. This might suggest that the thermal noise is irrelevant in the studied regime.
However, this is not the case. 
If the evolution time is large enough, so that $\rho(t)\approx \rho_\beta$, 
the encoded information cannot be retrieved from $\rho(t)$,
since $\rho_\beta$ does not depend on the initial state.
If a time $t\sim e^{\Omega(\beta^2)}$ has elapsed,
the system would have accommodated approximately $tL^3 e^{-\beta}\sim  e^{\Omega(\beta^2)}$ defects during the evolution.
This implies in particular that the system has endured $e^{\Omega(\beta^2)}$ errors which becomes significant for low temperatures.

Since Theorem~\ref{thm:main1} only provides a lower bound on the memory time, a natural question is whether
this bound is tight and, if so, what is the exact value of the constant coefficient $c$~?
To answer this question, the memory time of the 3D Cubic Code has been computed numerically
for a range of $\beta$'s and $L$'s. It should be emphasized that both Theorem~\ref{thm:main1} and our numerical simulation
use the same decoder at the readout step (see below).
Our simulation, described in detail in Section~\ref{sec:numerics}, provides a Monte Carlo estimate of the successful
decoding probability $p(t)$ as a function of the evolution time $t$. The decay of $p(t)$ was observed to follow
an exponential law, $p(t)=\exp(-t/\tau)$. The characteristic decay time $\tau$
was chosen as our numerical estimate of the memory time, that is, we set $T_{mem}=\tau$.
The scaling of $T_{mem}$ as a function of $L$ is shown on Fig.~\ref{fig:TvsL}.
 It suggests that $T_{mem}\approx L^{2.93 \beta-10.52}$
as long as $L\le L^*\approx e^{0.78\beta - 0.87}$. The numerical results strongly suggest that
our analytical bound Eq.~(\ref{Tmem-bound}) is tight up to constant coefficients (for the chosen decoder).
A comparison between the numerical data for $T_{mem}(\beta)$ and the analytical bound Eq.~(\ref{Tmem-bound1})
is shown on Fig.~\ref{fig:TvsBeta}. The numerics suggests a scaling $T_{mem}(\beta)\approx e^{1.69 \beta^2 +O(\beta)}$.
The numerical data were obtained for $26$ values of $\beta$ in the interval $4.0 \le \beta\le 5.25$ for each lattice size $5\le L \le 33$ satisfying $k(L)=2$ (that is, $L$ is odd and $L\ne 15$).
For each pair $(L,\beta)$ the memory time was estimated using a few hundred Monte Carlo samples.
Details on interpretation of the numerical data and the error analysis can be found in~Section~\ref{sec:numerics}.

The numerical computation of the memory time has become possible due
to our second result, which is an efficient error correction algorithm applicable to any topological stabilizer code.
We call it a  Renormalization Group (RG) decoder since it falls into a larger family
of error correction algorithms in which a syndrome is processed in a hierarchical way
using real-space renormalization methods~\cite{Harrington2004thesis,Duclos-CianciPoulin2009Fast}. Our algorithm and its rigorous analysis
borrow many ideas from~\cite{Harrington2004thesis}.
A detailed description of the RG decoder and its efficient implementation are described in
Section~\ref{sec:rgdecoder}. Let us briefly sketch the main idea of the algorithm.
The  RG decoding is a sequence of simple subroutines parameterized by integer
levels $p=0,1,\ldots,\log_2{L}$. At any given level $p$,
the decoder decomposes a syndrome
into disjoint connected clusters, where the connectivity is defined using $2^p$ as a unit of length.
This step can be implemented in time $O(N)$, where $N$ is the volume of the lattice.
The decoder then examines each cluster individually and tries
to `annihilate' it by a local Pauli operator. Clusters that cannot be annihilated are passed to the
next level (that is, $p+1$). The time needed to test whether
a given connected cluster can be annihilated depends on a particular code.
For the 3D Cubic Code we show how to perform this test in time  $O(V)$, where $V$ is the volume of the smallest rectangular
box enclosing the cluster. Once the decoder reaches the highest level, $p=\log_2{L}$, it returns the
product of annihilation operators over all clusters that have been successfully annihilated.
The overall running time of the RG decoder is $poly(N)$ for any topological stabilizer code.

The applicability of the RG decoder is by no means restricted to errors generated by thermal noise.
A natural question is how well the RG decoder performs against local stochastic noise.
In Appendix~\ref{sec:benchmark} we numerically benchmark the RG decoder for the 2D toric code,
where we consider random independent $X$ (or equivalently $Z$) errors with a uniform rate $p$.
Our numerical estimate of the error threshold is $p_c\approx 6.7\%$ which is reasonably close to the threshold $p_c\approx 11\%$ obtained using
the minimum weight matching decoding~\cite{DennisKitaevLandahlEtAl2002Topological}. Perhaps more importantly,
we prove that the RG decoder has a constant error threshold against local stochastic errors
for {\em any} topological stabilizer code, see Theorem~\ref{thm:universal-threshold} in
Appendix~\ref{sec:threshold}. To the best of our knowledge, the RG decoder is the first error correction
algorithm applicable to any topological stabilizer code that  combines good practical performance with
a rigorous lower bound on the error threshold and polynomial  worst-case running time.

\centerline{\bf Theorem~\ref{thm:main1}: sketch of the proof}

In Section~\ref{sec:thermal} we prove some basic facts about ``thermal errors" generated by
the dissipative Lindblad dynamics.
We define an energy barrier $\Delta(P)$ of a Pauli operator $P$ as the smallest
integer $m$ such that $P$ can be implemented by a stream of single-qubit errors
that does not create more than $m$ excitations during the process, see Section~\ref{subs:EnergyBarrier}.
The detailed balance condition and locality of quantum jump operators in the Lindblad
equation imply that thermal errors  are unlikely to have a high energy barrier.
More precisely, Lemma~\ref{lem:memory-time} in Section~\ref{subs:MemoryTime}
provides an upper bound on the storage error $\epsilon(t)$ assuming that any error with
energy barrier $\Delta(P)\le m$ will be corrected at the decoding step. Here $m$ is any fixed
``energy cutoff". The upper bound on $\epsilon(t)$ involves
the Boltzmann suppression factor $e^{-\beta m}$,
and an entropic factor that grows exponentially with the volume $N$ of the lattice.
We show that the entropic factor can be neglected in the regime $N e^{-\beta} \le 1$
which yields a bound
\begin{equation}
\label{sketch}
\epsilon(t)\le O(t) 2^{k(L)} N e^{-\beta m},
\end{equation}
see Section~\ref{subs:optimal}. Clearly, this bound is only useful if $m$ grows at least logarithmically with $N$.
At this point we must invoke the special properties of the 3D Cubic Code.
Our key technical ingredient is Theorem~2 of Ref.~\cite{BravyiHaah2011Energy}. It asserts (loosely) that
any Pauli operator $P$ creating a cluster of defects that contains a ``topologically charged'' sub-cluster
separated from its complement by distance $R$, has energy barrier at least $\Delta(P)\ge c\log{R}$ for some
constant coefficient $c$. Choosing the energy cutoff $m=c'\log{L}$ for some
constant $c'<c$, we can guarantee that a syndrome $S(P)$ caused by any Pauli operator $P$
satisfying $\Delta(P)\le m$,  consists of connected ``topologically neutral" clusters of defects,
where the connectivity is defined using $R\sim L^{c'/c}$ as a unit of length.
The RG decoder will identify each of these clusters as a connected component of the syndrome for
a sufficiently high RG level $p_0\approx \log_2{R}$. Since the clusters are topologically neutral, the RG decoder
will be able to annihilate each of them locally. This implies that the correcting
operator $P_{ec}(S)$ will annihilate all defects in $S$. Finally, we invoke the ``localization of high-level
errors lemma" (Lemma~2 in Ref.~\cite{BravyiHaah2011Energy}) to show that, in fact,
$P_{ec}(S)$ coincides with $P$ modulo stabilizers, see Section~\ref{subs:GoodErrors}.
Combining this observation and Eq.~\eqref{sketch}
we prove Theorem~\ref{thm:main1}.

It should be noted that the full hierarchy of the RG decoder is not needed for the proof of the theorem;
one can run the decoder at a single level ($p=p_0$), and skip all the other levels.
But, in practice, the hierarchical decoder performs better;
it is able to correct errors with slightly higher energy barriers.
Furthermore, the modified decoder using only one RG level would not have
a constant threshold against local stochastic errors, see Appendix~\ref{sec:threshold}.

\section{Previous work}
\label{sec:previous}

Here we briefly review alternative routes towards quantum self-correction
in topological memories proposed in the literature.
Most of them focus on  finding new mechanisms for suppressing diffusion of topological defects
(here and below we only consider zero-dimensional defects).
Arguably, the simplest of such mechanisms would be to have no topological defects in the first place.
Unfortunately, this seems to require four spatial dimensions.
The 4D toric code~\cite{DennisKitaevLandahlEtAl2002Topological} provides the only known example of a truly self-correcting quantum memory.
As was shown by Alicki and Horodecki's~\cite{AlickiHorodeckiHorodeckiEtAl2008thermal},
the memory time of the 4D toric code
grows exponentially with the lattice size for small enough bath temperature.
The first  3D topological memory in which diffusion of defects is constrained by superselection
rules was proposed by Chamon~\cite{Chamon2005Quantum}, see also~\cite{BravyiLeemhuisTerhal2010Topological}.
Topological defects in this model have a limited mobility restricted to certain
subspaces of $\mathbb{R}^3$ or have no mobility at all.
However, the model has no macroscopic energy barrier that could suppress the diffusion.
2D topological memories
in which diffusion of anyons is suppressed by effective long-range interactions
were studied by Chesi et al~\cite{ChesiRoethlisbergerLoss2009Self-Correcting}
and Hamma et al~\cite{HammaCastelnovoChamon2009ToricBoson}.
A quenched disorder and Anderson localization were proposed as
a means of suppressing propagation of defects at the zero temperature
by Wootton and Pachos~\cite{WoottonPachos2011disorder} and,
independently, by Stark et al~\cite{StarkImamogluRenner2011Localization},
see also~\cite{BravyiKoenig2011}.
A no-go theorem for quantum self-correction based on 3D stabilizer Hamiltonians
in which ground state degeneracy does not depend on lattice dimensions
was proved by Yoshida~\cite{Yoshida:2011}.
A different, very promising, line of research initiated
by Pastawski et al~\cite{PastawskiClementeCirac2011dissipation}
focuses  on quantum memories in which active error correction
is imitated by engineered dissipation driving the memory system
towards the ground state (as opposed to the Gibbs state).
Finally, let us emphasize that quantum  self-correction is technically different from
the thermal stability of topological phases,
see for instance~\cite{CastelnovoChamon2007Entanglement,
IblisdirEtAl2010thermal,
NussinovOrtiz2008Autocorrelations,
Hastings2011warmTQO}.
While the latter  attempts to establish the presence (or absence) of topological order
in the equilibrium thermal state, quantum self-correction is mostly concerned with the relaxation time
towards the equilibrium state.
\section{Analysis of thermal errors}
\label{sec:thermal}

We begin by setting up some terminology and notations.
A ground state of the memory Hamiltonian will be referred to as a {\em vacuum}.
It will be convenient to perform an overall energy shift such that the vacuum has zero energy.
A {\em Pauli operator} is an arbitrary tensor product of single-qubit Pauli operators
$X,Y,Z$ and the identity operators $I$. We will say that a Pauli operator creates
$m$ {\em defects} iff $P$ anticommutes with exactly $m$ stabilizer generators $G^X_c, G^Z_c$.
Equivalently, applying $P$ to the vacuum one obtains an eigenvector of $H$ with energy $m$.
For example,  using the explicit form of the generators, see Fig.~\ref{fig:CubicCode},
one can check that single-qubit $X$ or $Z$ errors create $4$ defects, while $Y$ errors create $8$ defects.
We will say that a Pauli error $P$ is corrected by the decoder iff $P_{ec}(S(P))=\gamma PG$,
where $S(P)$ is the syndrome of $P$, $G$ is a product of stabilizer generators  and $\gamma$ is an overall phase factor.
Let $N=2L^3$ be the total number of qubits. Note that $N$ is also the number of stabilizer generators.

\subsection{Energy barrier}
\label{subs:EnergyBarrier}

Let $\Gamma=(P_0,P_1,\ldots,P_t)$ be a finite sequence of Pauli operators
such that  the operators $P_i$ and $P_{i+1}$ differ on at most one qubit
for all $0\le i <t$. We say that $\Gamma$ is an {\em error path}
implementing a Pauli operator $P$ iff $P_0=I$ and $P_t=P$.
Let $m_i$ be the number of defects created by $P_i$.
The maximum number of defects
\[
m(\Gamma)=\max_{0\le i\le t} \; m_i
\]
will be called an {\em energy cost} of the error path $\Gamma$.
Given a Pauli operator $P$, we define its {\em energy barrier} $\Delta(P)$
as the minimum energy cost of all error paths implementing $P$,
\[
\Delta(P)=\min_\Gamma \; m(\Gamma).
\]
Although the set of error paths is infinite, the minimum always exists because the energy cost is a non-negative integer.
In fact, it suffices to consider paths in which $P_i$ are all distinct. The number of such paths is finite
since there are only finitely many Pauli operators for a given system size.

It is worth emphasizing that an operator $P$ may have a very large energy barrier even though
$P$ itself creates only a few defects or no defects at all.
Consider as an example the 2D Ising model, $H=-(1/2)\sum_{(u,v)} Z_u Z_v$, where the sum runs
over all pairs of nearest neighbor sites on the square lattice of size $L\times L$ with open boundary conditions.
Then the logical-$X$ operator $P=\bigotimes_u X_u$ has an energy barrier $\Delta(P)=L$ since any sequence of bit-flips implementing $P$ must create a domain wall across the lattice.

It is clear that a Pauli operator $P$ acting non-trivially on $n$ qubits has energy barrier at most $O(n)$.

\subsection{Lower bound on the memory time}
\label{subs:MemoryTime}

A naive intuition suggests that a stabilizer code Hamiltonian is a good candidate for being a self-correcting memory
if there exists an error correction algorithm, or a decoder, that corrects all errors with a sufficiently small energy
barrier. Errors with a high energy barrier   can confuse the decoder
and cause it to make wrong decisions,
but we expect that   such errors are unlikely to be created by the thermal noise. The following lemma makes this intuition more rigorous.
Let $f$ be the maximum energy barrier of Pauli operators that appear in the expansion of
the quantum jump operators $A_{\alpha,\omega}$ or $A_{\alpha,\omega}^\dag A_{\alpha,\omega}$,
see Section~\ref{subs:Davies}.
We will see later that $A_{\alpha,\omega}$ act on a constant number of qubits (see Proposition~\ref{prop:Lindblad} below).
It implies that $f=O(1)$. Below $m$ is an arbitrary energy cutoff.
\begin{lem}
\label{lem:memory-time}
Suppose an error correction algorithm $s \mapsto P_{ec}(s)$ corrects
any Pauli error $P$ with the energy barrier smaller than $m+2f$.
Let $\Phi_{ec}$ be the corresponding decoder defined by Eq.\eqref{decoder}.
Let $Q_m$ be the projector onto the subspace with at least $m$ defects. Then
\begin{equation}
\label{eq:infidelity}
\epsilon(t)\equiv \| \Phi_{ec}(\rho(t))-\rho(0) \|_1 \le
O(tN) \mathrm{Tr} \, Q_m e^{-\beta H}
\end{equation}
for any initial state $\rho(0)$ supported on the ground subspace of $H$.
The time evolution of $\rho(t)$ is governed by the Lindblad equation, Eq.~\eqref{eq:Markovian-master},
with the inverse temperature $\beta$ of the bath.
\end{lem}
\begin{proof}
We begin with the following simple observation. Let $\hat{H}$ be the linear map
taking the commutator with $H$, that is, $\hat{H}(X)=HX-XH$.
\begin{prop}
\label{prop:com}
The map $\hat{H}$ commutes with the Lindblad generator $\calL$.
The time-evolved state $\rho(t)$ commutes with the memory Hamiltonian $H$.
In addition, $\rho(t)=e^{\calL t}(\rho(0))$.
\end{prop}
\begin{proof}
Since $A_{\alpha,\omega}$ transfers energy $\omega$ from the memory to the bath, we have the following
identity:
\begin{equation}
\label{Acom}
A_{\alpha,\omega} H = H A_{\alpha,\omega} + \omega A_{\alpha,\omega} \quad \mbox{and} \quad
A_{\alpha,\omega}^\dag H = H A_{\alpha,\omega}^\dag - \omega A_{\alpha,\omega}^\dag.
\end{equation}
From this  one easily gets $\hat{H} \calL=\calL \hat{H}$. It follows that
\[
\rho(t)=e^{-i\hat{H} t + \calL t}(\rho(0)) =e^{\calL t}\circ  e^{-i\hat{H} t} (\rho(0))= e^{\calL t}(\rho(0)),
\]
since the initial state $\rho(0)$ commutes with $H$. It also implies that $\rho(t)$ commutes with $H$.
\end{proof}

Let $\calD = \ker \hat H$ be the set of all operators that are commuting with the Hamiltonian $H$.
Every operator in $\calD$ is block-diagonal in the energy eigenstate basis.
Since $\rho(0)$ is supported on the ground subspace of $H$, we have $\rho(0) \in \calD$.
Below we only consider states from $\calD$ and linear maps preserving $\calD$.
Define an orthogonal  identity decomposition
\[
I=\Pi_- + \Pi_+
\]
where $\Pi_-$ projects onto the subspace with energy $<m+f$
and $\Pi_+$ projects onto the subspace with energy $\ge m+f$.
Introduce auxiliary Lindblad generators
\begin{equation}
\calL_-( \rho)=\sum_\alpha \sum_\omega h(\alpha,\omega) \left(
B_{\alpha,\omega} \rho B_{\alpha,\omega}^\dag - \frac12 \{  \rho,B_{\alpha,\omega}^\dag B_{\alpha,\omega} \} \right),
\quad \mbox{where} \quad B_{\alpha,\omega} = \Pi_- A_{\alpha,\omega}
\label{Lp}
\end{equation}
and
\begin{equation}
\calL_+( \rho)=\sum_\alpha \sum_\omega h(\alpha,\omega) \left(
C_{\alpha,\omega} \rho\,  C_{\alpha,\omega}^\dag - \frac12 \{  \rho,C_{\alpha,\omega}^\dag C_{\alpha,\omega} \} \right),
\quad \mbox{where} \quad C_{\alpha,\omega} = \Pi_+ A_{\alpha,\omega}.
\label{Lq}
\end{equation}
Simple algebra shows that $\calL_-$ and $\calL_+$ preserve $\calD$ and
their restrictions on $\calD$ satisfy
\begin{equation}
\label{LPQ}
\calL=\calL_- + \calL_+.
\end{equation}
It is useful to note that any $X \in \calD$ commutes with $\Pi_\pm$.
By abuse of notations, we shall
apply Eq.\eqref{LPQ} as though it holds for all operators.
\begin{prop}
\label{prop:Lp}
For any time $t\ge 0$ and for any state $\rho_0$ supported on the ground subspace of $H$ one has
\begin{equation}
\Phi_{ec}(e^{\calL_- t}(\rho_0))=\rho_0.
\end{equation}
\end{prop}
\begin{proof}
Since all maps are linear, we may assume $\rho_0 =\ket{g}\bra{g}$ is a pure state.
Then, $e^{\calL_- t}(\rho_0)$ is in the span of states of form $\ket{\psi}=\Pi_- E_n \cdots \Pi_- E_2 \Pi_- E_1 \ket{g}$,
where $E_i$ are Pauli operators that appears in the expansion of $A_{\alpha,\omega}$
or $A_{\alpha,\omega}^\dag A_{\alpha,\omega}$.
This follows from the Taylor expansion of $e^{\mathcal{L}_- t}$.
Since Pauli errors map eigenvectors of $H$ to eigenvectors of $H$, we conclude that
either $\ket{\psi}=0$, or $\ket{\psi}=E_n \cdots E_2 E_1 \ket g$.
Furthermore, the latter case is possible only if the Pauli operator $E\equiv E_n \cdots E_2 E_1$
has energy barrier smaller than $m+2f$. Indeed, definition of $\Pi_-$ implies that $E_j E_{j-1} \cdots E_1$
creates at most $m+f-1$ defects for all $j=1,\ldots,n$. By assumption, each operator $E_j$ can be implemented by an error
path with energy cost at most $f$. Taking the composition of all such error paths
one obtains an error path for $E$ with energy cost at most $m+2f-1$
and thus  $\Phi_{ec}$ will correct $E$.
Since $\Phi_{ec} \  \circ \  e^{\calL_- t}$ is a TPCP map, we must have $\Phi_{ec}(e^{\calL_- t}(\rho_0))=\rho_0$.
\end{proof}
\begin{prop}
\label{prop:exp}
For any decomposition $\calL=\calL_- + \calL_+$ one has the following identity:
\begin{equation}
e^{\calL t} = e^{\calL_- t} + \int_0^t ds \, e^{\calL_- (t-s)} \calL_+ \, e^{\calL s}.
\label{exp}
\end{equation}
\end{prop}
\begin{proof}
Use the identity
\[
\frac{d}{ds} e^{\calL_-(t-s)} e^{\calL s} = e^{\calL_- (t-s)} (-\calL_- + \calL) e^{\calL s} =  e^{\calL_- (t-s)} \calL_+ \, e^{\calL s}.
\]
\end{proof}
\begin{prop}[\cite{ChesiLossBravyiEtAl2009Thermodynamic}]
\label{prop:Lindblad}
Each operator $A_{\alpha,\omega}$ acts non-trivially only on $O(1)$ qubits.
Furthermore,
\[
 \| \calL_+ \|_1 \equiv \sup_X \frac{ \| \calL_+(X) \|_1 }{ \| X \|_1 }  = O(N).
\]
\end{prop}
\begin{proof}
Let $A_{\alpha}(t)=e^{iHt}A_\alpha e^{-iHt}$ such that $A_{\alpha,\omega}$
is the Fourier component of $A_\alpha(t)$ with the Bohr frequency  $\omega$,
that is, $A_\alpha(t)=\sum_\omega e^{-i\omega t} A_{\alpha,\omega}$.
Recall that each operator $A_\alpha$ acts on $O(1)$ qubits.
Since $H$ is a sum of pairwise commuting terms, we can represent
$A_\alpha(t)$ as  $A_\alpha(t)=e^{iH_\alpha t} A_\alpha e^{-iH_\alpha t}$
where $H_\alpha$ is obtained from $H$ by retaining only those stabilizer generators that do not
commute with $A_\alpha$. All such generators must share at least one qubit with $A_\alpha$.
Therefore  $A_\alpha(t)$ and $A_{\alpha,\omega}$ act non-trivially only $O(1)$ qubits.
Furthermore, since $H_\alpha$ has $O(1)$ distinct integer eigenvalues,
$A_\alpha(t)$ has only $O(1)$ distinct Bohr frequencies.
The bound $\| A_{\alpha,\omega} \|\le 1$ follows trivially from our assumption $\| A_\alpha\|\le 1$.
The norm of $\mathcal{L}_+$ is then
bounded from Eq.~\eqref{eq:lindbladian} using triangle inequality, $\| X Y \|_1 \le \| X \| \cdot \| Y \|_1$, and Eq.~\eqref{eq:weak}.
\end{proof}

Now we are ready to bound $\epsilon(t)$. Let $\rho_0\equiv \rho(0)$.
First, using Proposition~\ref{prop:com} one gets $\rho(t) = e^{\calL t} (\rho_0)$.
Applying Propositions~\ref{prop:Lp},\ref{prop:exp} one arrives at
\begin{equation}
\epsilon(t)\le \int_0^t ds\, \| e^{\calL_-(t-s)}  \calL_+ \, e^{\calL s} (\rho_0) \|_1
\le t \cdot \max_{0\le s\le t} \; \|  \calL_+ \, e^{\calL s} (\rho_0) \|_1.
\end{equation}
We shall use an identity
\begin{equation}
\calL_+(X)=\calL_+( Q_m X Q_m)
\end{equation}
valid for any $X \in \calD$. Indeed,
any quantum jump operator $A_{\alpha,\omega}$ changes the energy at most by $f$,
so that $\calL_+(Q_m^\perp X)=\calL_+(XQ_m^\perp)=0$  for any $ X \in \calD$ (note that any $ X \in \calD$
commutes with $Q_m$).
We arrive at
\begin{equation}
\epsilon(t)
\le t \cdot \| \calL_+(Q_m e^{\calL s} (\rho_0) Q_m) \|_1
\le t \cdot \| \calL_+ \|_1 \cdot \| Q_m e^{\calL s} (\rho_0)  Q_m \|_1
\le O(tN) \trace{ Q_m  e^{\calL s} (\rho_0)},
\label{epsilon_bound}
\end{equation}
where the maximization over $s$ is implicit.
Since the ground state energy of $H$ is zero, one has
\begin{equation}
\rho_0=\calZ_\beta \rho_\beta \rho_0,
\end{equation}
where $\calZ_\beta$ is the partition function. It yields
\begin{equation}
\trace{ Q_m  e^{\calL s} (\rho_0)}
= \trace{  \rho_0 \, e^{\calL^* s}(Q_m)}
= \calZ_\beta \trace{ \rho_\beta \rho_0 e^{\calL^* s} (Q_m)}
= \calZ_\beta \langle \rho_0, e^{\calL^* s} (Q_m) \rangle_\beta,
\end{equation}
where we introduced the Liouville inner product $\langle X,Y\rangle_\beta \equiv \trace{\rho_\beta X^\dag Y}$.
We denoted by $\calL^*$ the linear map adjoint to $\calL$ (the one describing time-evolution in the Heisenberg picture) with respect to the Hilbert-Schmidt inner product.
The detailed balance condition Eq.~\eqref{DB} implies that
the linear map $\calL^*$ is self-adjoint with respect to the Liouville inner product,
\begin{equation}
\langle X, \calL^*(Y)\rangle_\beta= \langle \calL^*(X),Y\rangle_\beta,
\label{eq:liouville-innerproduct}
\end{equation}
see~\cite{AlickiHorodeckiHorodeckiEtAl2008thermal}. It can be easily checked using the identities
\begin{equation}
\label{id1}
A_{\alpha,\omega} \rho_\beta = e^{-\beta \omega} \rho_\beta A_{\alpha,\omega} \quad \mbox{and} \quad
A_{\alpha,-\omega} = A_{\alpha,\omega}^\dag
\end{equation}
and Eq.~(\ref{DB}). It follows that the map $e^{\calL^* s}$ is also self-adjoint with respect to the Liouville inner product. Hence, we have
\begin{equation}
\trace{ Q_m  e^{\calL s} (\rho_0)}
=   \calZ_\beta \langle \rho_0, e^{\calL^* s} (Q_m) \rangle_\beta
=   \calZ_\beta \langle e^{\calL^* s} (\rho_0), Q_m \rangle_\beta
=   \trace{e^{\calL^* s}(\rho_0) Q_m e^{-\beta H}}
\le \trace{Q_m e^{-\beta H}}.
\label{trace_final}
\end{equation}
Here the last inequality follows from the fact that $e^{\calL^* s}$ is a unital CP map
and from $\rho_0 \le I$.
\end{proof}

\subsection{Choosing the optimal number of qubits}
\label{subs:optimal}

Let us fix the inverse temperature $\beta$ and ask what is the optimal system size $N$
that minimizes the upper bound on the storage error $\epsilon(t)$ derived in
Lemma~\ref{lem:memory-time}.
The dimension of the subspace with exactly $n$ defects equals $2^{k(L)}$ times
the  binomial $N$ choose $n$ coefficient since the total number of potential defects locations is $N$
(recall that the 3D Cubic Code has exactly $N$ stabilizer generators).
Choosing any constant $0<a<1$ and setting $k\equiv k(L)$ we obtain
\begin{equation}
\label{eq:tradeoff}
\epsilon(t)
\le O(t) N 2^{k} e^{-a\beta m}  \sum_{n\ge m} \binom{N}{n} e^{-(1-a)\beta n}
\le O(t) N 2^k e^{-a\beta m}  (1+e^{-(1-a)\beta})^N.
\end{equation}
The Boltzmann factor $e^{-a\beta m}$ is responsible for the self-correcting behavior. It supports our initial
intuition that errors with a high energy barrier (at least $m$) are unlikely to be generated by the thermal noise.
On the other hand, the factor $(1+e^{-(1-a)\beta})^N$ represents the entropy contribution. Loosely speaking, it takes into
account the  number of error paths with a given energy cost. The optimal choice of $N$ is determined by the
tradeoff between the  Boltzmann factor and the entropy factor which depends on the scaling of the energy barrier
$m$. For example, suppose one chooses $N\le e^{(1-a)\beta}$ such that the entropy factor is upper bounded by
a constant. Then Eq.~\eqref{eq:tradeoff} yields
\begin{equation}
 \epsilon(t)\le  O(t) N 2^k e^{-a\beta m} \quad \mbox{for all $N\le e^{(1-a)\beta}$}.
\label{eq:tradeoff1}
\end{equation}
In Section~\ref{sec:Correctability}, we will argue that the 3D Cubic Code
(and any other stabilizer code having the topological order
and obeying the no-strings rule as stated in~\cite{BravyiHaah2011Energy}) admits an efficient error correction algorithm that
corrects all errors with the energy barrier smaller than $m = c \log{(N)}$ for some constant $c>0$.
For such codes one gets
\begin{equation}
\epsilon(t)\le  O(t) \cdot 2^k \cdot N^{1-ac \beta} \quad \mbox{for all $N\le e^{(1-a)\beta}$}.
\label{eq:tradeoff2}
\end{equation}
This is equivalent to the statement of Theorem~\ref{thm:main1}.
\section{Renormalization group decoder}
\label{sec:rgdecoder}

Recall that the syndrome measurement reveals locations of defects (flipped stabilizer generators)
created by an unknown error. The renormalization group (RG) decoder attempts to annihilate the defects comprising
the syndrome $S$ by dividing them into disjoint connected clusters $S=C_1\cup \ldots \cup C_m$ and then trying to annihilate each cluster $C_a$  individually.
More specifically, the decoder checks whether $C_a$ can be annihilated by a Pauli operator  $P_a$
supported  on a sufficiently small spatial region  $b(C_a)$ enclosing $C_a$.
If such local annihilation operator $P_a$ exists, the decoder  updates the syndrome by erasing
all the defects comprising $C_a$, records the operator $P_a$, and moves on to the next cluster.
If $C_a$ cannot be annihilated, the decoder skips it.
In general, the annihilation operator $P_a$ might not be unique.
However, if the enclosing region $b(C_a)$ is small enough to ensure that no logical operator
can be supported on $b(C_a)$, all annihilation operators $P_a$ must be equivalent modulo stabilizers
and the choice of $P_a$ does not matter.

After all clusters $C_a$ have been examined, the decoder is left with a new configuration of defects
$S'$, which is typically smaller than the original one.
If no defects are left, i.e., $S'=\emptyset$, the decoder stops and returns the product of all recorded Pauli operators $P_a$.
If $S'\ne \emptyset$,  the decoder applies a scale transformation
increasing the unit of length by some constant factor and repeats all the above steps starting from the syndrome $S'$.
The scale transformation potentially merges several unerased clusters $C_a$ into a single connected cluster
whereby giving the decoder one more attempt to annihilate them.

The full decoding algorithm is the iteration of partitioning the defects into the connected clusters and calculating the annihilation operators. It declares failure and aborts if the recorded operator cannot annihilate all the defects before the rescaled unit length is comparable to the lattice size.

A detailed  implementation of the RG decoder must  be tailored to a specific lattice geometry
and  a stabilizer code under consideration.
It must include a precise definition of  the connected clusters of defects $C_a$
and the enclosing regions $b(C_a)$. It must also include an algorithm for choosing the annihilation
operators $P_a$, a schedule for increasing the unit of length,
and clearly stated conditions under which the decoder aborts.
In the rest of this section we describe an efficient implementation of the RG decoder
for arbitrary 3D stabilizer codes satisfying topological order conditions defined in~\cite{BravyiHaah2011Energy}.
The only part of this implementation specialized for the 3D cubic code is the ``broom algorithm" of Section~\ref{subs:broom}.

\subsection{Metric convention and terminology}

Let $\Lambda$ be the regular 3D cubic lattice of linear size $L$ with periodic boundary conditions
along all coordinates $x,y,z$.
We shall label sites of $\Lambda$ by triples of integers $(i,j,k)$ defined modulo $L$
and measure the distance between sites using the $\ell_\infty$-metric. In other words,
the distance $d(u,v)$ between a pair of sites $u$ and $v$ is the smallest
integer $r$ such that $u$ and $v$ can be enclosed by a cubic box with dimensions
$r\times r\times r$. For example, $d(u,v)=1$ whenever $u$ and $v$ belong to the same
edge, plaquette, or elementary cube of the lattice.
Each site of $\Lambda$ represents one or several physical qubits (two qubits for the 3D Cubic Code).
Each elementary cube $c$ represents a spatial location of one or several  stabilizer generators
(two generators for the 3D Cubic Code). A generator located at cube $c$ may act only on
qubits located at vertices of $c$. We shall label each elementary cube by
coordinates of its center, the triple of half-integers $(i,j,k)$ defined modulo $L$.
The distance $d(c,c')$ between a pair of cubes $c$ and $c'$ is the
distance between their centers.    For example, $d(c,c')=1$ whenever $c$ and $c'$ share a vertex, an edge, or a plaquette.

A defect is a stabilizer generator whose eigenvalue has been flipped as a result of the error.
We shall use a term {\em cluster of defects}, or simply {\em cluster}
for any set of defects.  Define the {\em diameter} of a cluster $d(C)$ as
the maximum distance $d(c,c')$ where $c,c'\in C$.
Here and below  the distance between defects is defined as the distance between the cubes
occupied by these defects.
Given two non-empty clusters $C$ and $C'$,
define a distance $d(C,C')$ as the minimum distance $d(c,c')$ where $c\in C$ and $c'\in C'$.
Given an integer $r$, we shall say that a cluster $C$ is \emph{connected at scale $r$},
or simply \emph{$r$-connected}, if $C$
cannot be partitioned into two proper subsets $C=C'\cup C''$ such that $d(C',C'')>r$.
A maximal $r$-connected subset of a cluster $C$  is called a {\em $r$-connected component} of $C$.
The {\em minimal enclosing box} $b(C)$ of a cluster $C$ is the smallest
rectangular box $B$ enclosing all defects of $C$ such that all vertices of $B$ are dual sites
of $\Lambda$.  Note that the minimal enclosing box $b(C)$ is unique as long as $d(C)<L/2$
(if a cluster $C$ has diameter $L/2$,  one may have two boxes with the same dimensions  enclosing $C$ that `wrap' around the lattice in two different ways).

\subsection{Topological stabilizer codes}

The RG decoder can be applied to any stabilizer code satisfying topological order
conditions defined in~\cite{BravyiHaah2011Energy}. For the sake
of completeness, we state these conditions below. Let $\cal G$ be the Abelian group
generated by the stabilizer generators. Elements of $\cal G$ will be called stabilizers.
Let $S(P)$ be the syndrome of a Pauli operator $P$, that is,
the set of all stabilizer generators anticommuting with $P$.
The syndrome can be viewed as a cluster of defects.
\begin{defn}
\label{dfn:TQO}
A stabilizer code has \emph{topological order at scale $\ltqo$} if any Pauli operator supported on a cube of linear size $\ltqo$
satisfies the following conditions:\\
(1) If $P$ commutes with all stabilizers then $P$ is stabilizer itself;\\
(2) If $S(P)\ne \emptyset$ then there is a stabilizer $G\in \cal G$  such that $PG$ is supported on the $1$-neighborhood of
the minimum enclosing box $b(S(P))$.\\
A family of stabilizer codes $\{\mathcal{C}(L)\}_L$ is called topological if there exists a constant $\gamma>0$ such that each code $\mathcal{C}(L)$ has topological order at scale $\ltqo \ge L^\gamma$.
\end{defn}
The 3D Cubic Code satisfies this topological order condition with
$\ltqo = \half L$, see~\cite{Haah2011Local}.
Strictly speaking $\ltqo = L-2$ for the 3D Cubic code,
but in order to avoid unnecessary complications due to boundaries,
we always assume that $\ltqo \le \half L$.
Below we consider only topological stabilizer codes.
\begin{defn}
 A cluster of defects $C$ is called \emph{neutral} if it can be created from the vacuum by a Pauli operator $P$ supported on a cube of linear size $\ltqo$. Otherwise, the cluster is said to be \emph{charged}.
\end{defn}
For example, the 2D toric code~\cite{Kitaev2003Fault-tolerant} has two types of defects: magnetic charges (flipped plaquette operators) and electric charges (flipped star operators).
A cluster of defects $C$ is neutral iff $C$ contains even number of magnetic charges and even number of electric charges.

It follows from Definition~\ref{dfn:TQO}
that any neutral cluster of defects $C$ can be annihilated by a Pauli operator supported on the $1$-neighborhood
of the minimum enclosing box $b(C)$.

\subsection{RG decoder}
\label{subs:rgdecoder}

We are now ready to define our RG decoder precisely.
Recall that $d(C)$ is the diameter of a cluster $C$, and $\ltqo \le \half L$ by convention.
\begin{center}
\fbox{
\parbox{16cm}{{\bf TestNeutral}
\begin{flushleft}
{\bf Input} $S$ : a set of defects, {\bf Output} $P$ : a Pauli operator.\\
1. Compute the minimal enclosing box $B$ of $S$.\\
2. {\bf if} $d(B) > \ltqo$, {\bf then} {\bf return} $I$.\\
3. Try to compute a Pauli $P$ supported on the 1-neighborhood of $B$ such that $S(P) = S$.\\
4. {\bf if} a consistent $P$ is found {\bf then} {\bf return} $P$ {\bf else} {\bf return} $I$.
\end{flushleft}
}}
\end{center}
The topological order condition of Definition~\ref{dfn:TQO} implies that
TestNeutral computes the correcting Pauli operator for any neutral cluster.
Step~1 is easy as we discuss in the end of Section~\ref{subs:clusters}.
The specification of Step~3 depends on the code, but it always has an efficient implementation
using the standard stabilizer formalism~\cite{Gottesman1998Theory}.
In general, the condition $S(P)=S$ can be described by
a system of $O(V)$ linear equations over $O(V)$ binary variables parameterizing $P$,
where $V$ is the volume of $B$.
The running time is then $O(V^3)$ by the Gauss elimination.
In the special case of the 3D Cubic Code,
there is a much more efficient algorithm running in time $O(V)$
which we describe in Section~\ref{subs:broom}.

Let $p_{M}$ be the largest integer such that $2^{p_M} < \min(\half L,\ltqo)$.
For any integer $0\le p\le p_M$, we define
the \emph{level-$p$ error correction} {\bf EC($p$)}:
\begin{center}
\fbox{
\parbox{16cm}{{\bf EC($p$)}
\begin{flushleft}
{\bf Input} $S$ : a syndrome, {\bf Output} $P$ : a Pauli operator.\\
1. Partition $S$ into $2^p$-connected components: $S=C_1\cup \ldots \cup C_m$.\\
2. For each component, compute $P_a = $ TestNeutral$(C_a)$.\\
3. {\bf return} the product $P_1 P_2 \cdots P_m$.
\end{flushleft}
}}
\end{center}
The overall running time of EC($p$)  is polynomial in the number of qubits $N$.
Step~1 can be done, for example, by examining the distance between all pairs of defects,
forming a graph whose edges connect pairs of defects separated by distance $\le 2^p$, and finding
connected components of this graph.
A more efficient algorithm with the running time $O(N)$
is described in Section~\ref{subs:clusters}.
Since $V=O(N)$ and $m=O(N)$,
we see that the worst-case running time of EC($p$) is $O(N^2)$.
For instance, one can consider nested boxes,
near the faces of which many defects lie.
However, clusters are created from the vacuum
with some probability which we expect to be smaller than, say, $\half$.
So the number of clusters that have overlapping minimal enclosing box
appears with exponentially small probability.
On average, the running time of EC($p$) will be $O(N)$.

The full RG decoding algorithm is as follows.
\begin{center}
\fbox{
\parbox{16cm}{{\bf RG Decoder}
\begin{flushleft}
{\bf Input} $S$ : the syndrome, {\bf Output} $P_{ec}$ : a Pauli operator.\\
1. Set $P_{ec}=I$.\\
2. \parbox[t]{15cm}{{\bf for} $p=0$ {\bf to}  $p_M$ {\bf do}\\
\makebox[5mm]{}  \parbox[t]{14cm}{ Let $Q=$EC($p$)($S$).}\\
\makebox[5mm]{}  \parbox[t]{14cm}{Update $P_{ec} \leftarrow P_{ec} \cdot Q$ and $S\leftarrow S\oplus S(Q)$.}\\
\parbox[b]{5cm}{\bf end for}} \\

\vspace{3mm}

3. {\bf if} $S=0$ {\bf then} return $P_{ec}$ {\bf else} declare failure.
\end{flushleft}
}}
\end{center}
Here the notation $S\oplus S(Q)$ stands for the symmetric difference of the sets
$S$ and $S(Q)$ or addition modulo two, if the syndromes are represented by binary strings.
To conform with conventions made in Section~\ref{subs:decoding},
we may let the decoder on failure return an arbitrary Pauli operator
that is consistent with the measured syndrome.
Of course, this is no better than initializing the memory with an arbitrary state.
The discussion above implies that the RG decoder has running time $O(N^2\log{N})$ in the worst case,
and $O(N\log{N})$ in the case of sparse syndromes.

\subsection{Cluster decomposition}
\label{subs:clusters}
\newcommand{\ls}{r}

Given a length scale $\ls$, the cluster decomposition of defects
is to partition the defects into maximally connected subsets
(connected components) of the syndrome at scale $\ls$.
Naively, the task to compute the decomposition of all the defects into clusters
will take time $O(m^2)$ where $m$ is the total number of defects.
The density of defects will typically be constant irrespective of the system size,
and the computation time for decomposition will be $O(N^2)$,
where $N$ is the volume of the system.
However, by exploiting the geometry of simple cubic lattice,
we can do it in time $O(N)$.
This is the optimal scaling
since we have to sweep through the whole system anyway
to identify the position of defects.

If $\ls = 1$, the problem is to label the connected components of binary array \cite{AsanoTanaka2010,KiranEtAl2011Labeling}.
Given a defect $u_0$, we can compute the connected component
containing $u_0$ in time $O(m)$ where $m$ is the number of defects in the component.
One prepares an empty queue (first-in-first-out data structure),
and puts $u_0$ into it.
The subsequent computing is as follows:
(i) Pop out the first element $u$ from the queue,
and of the neighborhood put the \emph{unlabeled} defects into the queue and label them.
(ii) Repeat until the queue becomes empty.
Every defect in a connected component $j$ is stored in the queue only once.
Hence, this process computes the component $j$ of a given defect
in time proportional to the number $m_j$ of defects in $j$.
One examines the whole system in some order
and finds the connected component whenever there is an unlabeled defect.
The total computation time is proportional to $N + O(1)\sum_j m_j = O(N)$,
since the connected components are disjoint.

For $\ls > 1$, the algorithm begins by dividing the whole lattice
into boxes of linear size $\ls$ or smaller.
The defects in a box certainly belong to a single connected component
(recall that we use the $l_\infty$ metric).
The defects in the boxes $B, B'$ belong to the same component
if and only if there is a pair $u \in B$, $v \in B'$ of defects
such that $d(u,v) \le \ls$.
In other words, we evaluate the binary function
\[
 \delta(B, B') =
\begin{cases}
1 & \text{if there are $u \in B, v \in B'$ such that } d(u,v) \le \ls ,\\
0 & \text{otherwise}
\end{cases}
\]
for each neighbor $B'$ of $B$;
if $B$ does not meet $B'$,
we know that $\delta(B,B') = 0$.

Given the table of $\delta$,
we can finish computing the decomposition
in time $O((L/r)^D)$ as in the $\ls = 1$ case.
We show that the computation of $\delta(B,B')$ can be done
in time $O(r^D)$ where $D$ is the dimension of the lattice.
Then, the total time to compute the table of $\delta$ for all adjacent boxes
will be $O(r^D (L/r)^D )=O(N)$.
Let $B$ and $B'$ be adjacent.
For clarity of presentation, we restrict to $D=2$.
Suppose $B$ and $B'$ meet along an edge parallel to $x$-axis.
Since any difference $|x-x'|$ of $x$-coordinates
of the defects in $B\cup B'$ is at most $\ls$,
we only need to compare $y$-coordinates.
That is, the problem is reduced to one dimension.
It suffices to pick two defects from $B,B'$ respectively
that are the closest to the $x$-axis.
If the $y$-coordinates differ at most by $\ls$,
then $\delta(B,B') = 1$;
otherwise, $\delta(B,B')=0$.

Suppose $B$ and $B'$ meet at a vertex.
Without loss of generality, we assume $B$ is in the third quadrant,
and $B'$ is in the first quadrant.
Define a binary function $\delta'$ on $B'$ as
\[
 \delta'(i,j)=
\begin{cases}
 1 & \text{if there is a defect $(x,y) \in B'$ where $ x \le i$ and $ y \le j$}, \\
 0 & \text{otherwise}.
\end{cases}
\]
The function table of $\delta'$ is computed in time $O(\ls^2)$.
It is important to note that $\delta'(i,j) = 1$ implies $\delta'(i+1,j) = \delta'(i,j+1)=1$.
One starts from the origin and sets $\delta'(\half,\half)=1$
if there is a defect at $(\half,\half)$; otherwise $\delta'(\half,\half)=0$.
Here, $(\half,\half)$ means the elementary square in the first quadrant
that is the closest to the origin.
Proceeding by a lexicographic order of the coordinates,
one sets $\delta'(i,j) = 1$
if $\delta'(i-1,j)=1$, or $\delta'(i,j-1)=1$, or there is a defect at $(i,j)$;
otherwise $\delta'(i,j)=0$.
It is readily checked that this procedure correctly computes $\delta'$.
Equipped with this $\delta'$ table,
we can immediately test for each defect in $B$
whether there is a defect in $B'$ within distance $\ls$.
Thus, we have computed $\delta(B,B')$ in time $O(r^2)+O(m)$
where $m$ is the number of defects in $B$,
which is at most $O(r^2)$.
The computation of $\delta$ in higher dimensions is similar.

The computation of the minimal enclosing box for each cluster is also efficient.
Given the coordinates of the $m$ points in the cluster,
we read out, say, $x$-coordinates $x_1,\ldots,x_m$.
The minimal enclosing interval $B_x$ of $x_1,\ldots,x_m$ under periodic boundary conditions,
is the complement of the longest interval between consecutive points $x_i$ and $x_{i+1}$
which can be computed in time $O(m)$.
$B_x$ is unambiguous if the diameter of the cluster is smaller than $L/2$.
The minimal enclosing box is the product set $B_x \times B_y \times B_z$,
whose vertices are computed in time $O(m)$.

\subsection{Broom algorithm}
\label{subs:broom}

Now we describe an efficient algorithm for the 3D cubic code
that tests whether a cluster is neutral. If the test is positive,
the algorithm also returns a
Pauli operator $E$ that annihilates the cluster.
Recall that the 3D cubic code has property
that if the cluster is neutral, $E$ is supported in the minimal box that encloses the cluster.
This is because the erasure process \cite{Haah2011Local} is possible
as long as it does not encounter a defect.

\begin{figure}
\centering
\includegraphics[width=.4\textwidth]{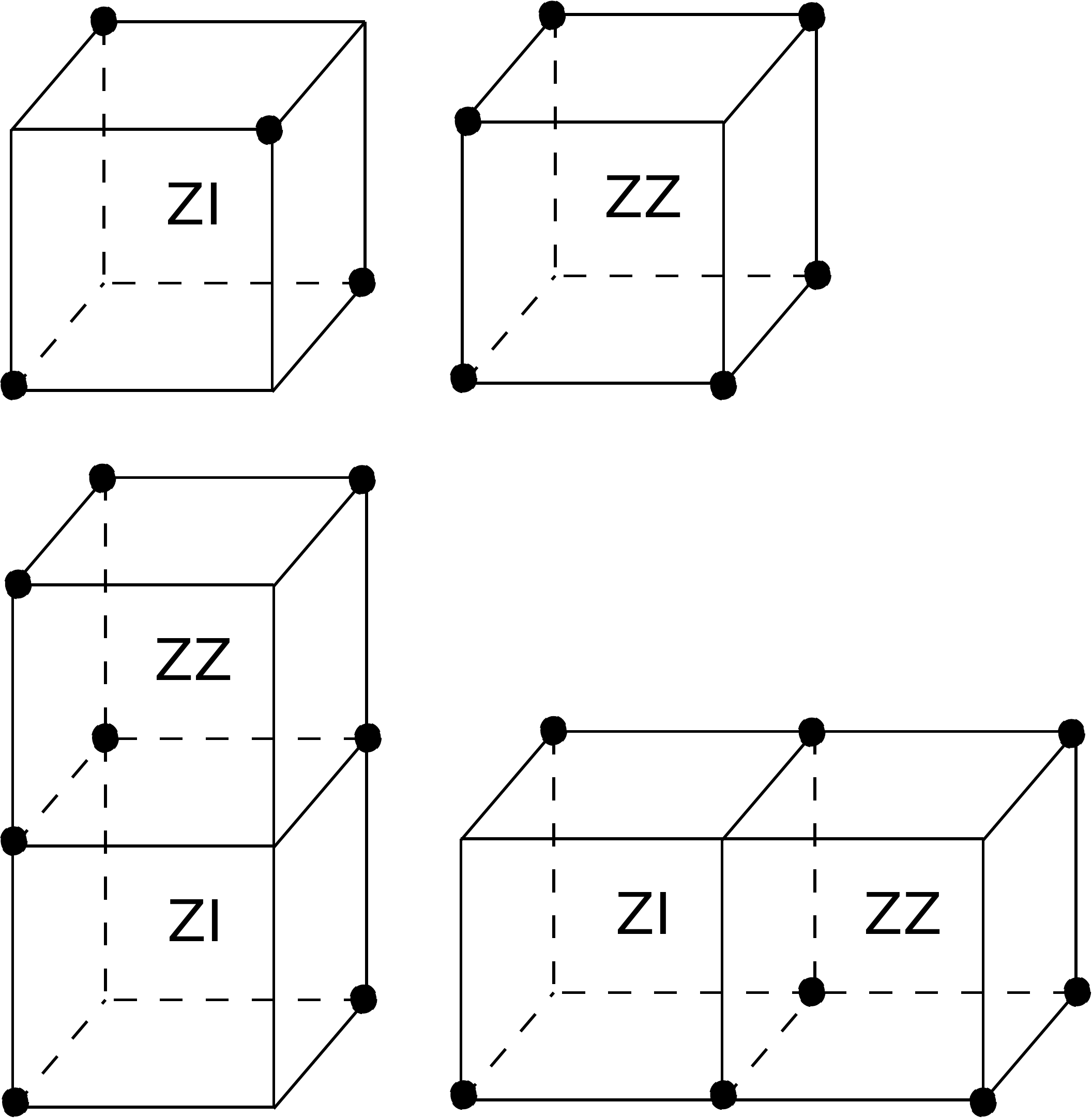}
\caption{Elementary syndromes created by $Z$ errors. The vertices which are on the dual lattice, represent the defects created by the error at the center. The elementary syndrome by $ZI$ is used to push the defects to the bottom and to the left. The syndromes by errors of weight three is used to push the defects to the bottom-left corner.}
\label{fig:elementary-syndrome}
\end{figure}

By the duality between $Z$- and $X$-type operator,
we assume $E$ is consisted entirely of $Z$ errors.
Since in the cubic code there is only one $X$-type of stabilizer generators,
the positions of defects completely determines the $X$-type syndrome.
We shall use double-letter indices to represent single-site Pauli operators;
for example, $IZ \equiv I \otimes Z$, $ZZ \equiv Z \otimes Z$, $II \equiv I \otimes I$, etc.
Let the qubits have integer coordinates, and the defects have half-integer coordinates.

Elementary syndromes are depicted in Fig.~\ref{fig:elementary-syndrome}.
We fix a box $B$ that encloses all the defects in the neutral cluster.
We will sweep the defects to bottom-left-back corner.
Since each defect is a $Z_2$ charge, they will disappear in the end.
The algorithm begins with the top-right foremost vertex of $B$ on the dual lattice.
If there is a defect at $(x+\half,y+\half,z+\half)$,
we apply $ZI$ at $(x,y,z)$ to eliminate it.
This might create another defects as there are four defects in the elementary syndrome.
Important is that the potentially new defects
are all contained in the box $B$ we started with.
Continuing with $ZI$ we push all the defects in the foremost plane of $B$
to the left vertical line and the bottom horizontal line.
During this process, we record where $ZI$ has been applied.

For the defects on the vertical line at the left or on the horizontal line at the bottom,
we use the operator of weight 3 to further move the defects to the bottom-left corner.
See Fig.~\ref{fig:elementary-syndrome}.
This will in general create more defects behind, all of which are still contained in $B$.
Thus, we have moved all the defects on the foremost plane to the bottom-left corner
except for the three sites $t,u,v$:
\begin{equation}
\xymatrix@!0{ t \ar@{-}[d] & o \\ u \ar@{-}[r] & v }
\label{eq-diagram:last-three-defects}
\end{equation}
That is, if $E'$ is the recorded operator during the sweeping process,
the syndrome $S(EE') \subseteq B$ has potential defects only at $t,u,v$ on the foremost plane.

Let $o$ be at $(x_o+\half,y_o+\half,z_o+\half)$.
By considering the multiplication by suitable stabilizer generators $Q^Z$,
we can assume that $EE'$ is the identity on the plane $x=x_o$, except $(x_o,y_o,z_o)$.
Since there is no defect at $o$,
the operator at $(x_o,y_o,z_o)$ has to commute with $XX$;
it is either $II$ or $ZZ$.
Applying $ZZ$ if necessary, the operator at $(x_o,y_o,z_o)$ will become $II$,
and the defects at $t,u,v$ will disappear.
In this way, we have successfully pushed all the defects to the next-to-foremost plane.
We emphasize that the box $B$ still envelops all the defects,
and further $B$ can be shrunk in one direction.

Due to the threefold symmetry of the cubic code,
one can carry out this broom algorithm along any of three directions.
We will have at last a box $B$ of volume 1 that encloses all defects.
The defects in the cluster must be from one of the three elementary syndrome cubes
created either by $ZI$, $IZ$, or $ZZ$, which are easily eliminated.
It is clear that in time $O(V)$ the error operator has been computed up to stabilizer,
where $V$ is the initial volume of the minimal enclosing box of the cluster.

Note that a similar sweeping algorithm
is applicable to the toric code family for any dimensions.

\section{Correctability of errors with a logarithmic energy barrier}
\label{sec:Correctability}

Let $P$ be an unknown Pauli error. Suppose we are promised that $P$ has a sufficiently small
energy barrier, namely, $\Delta(P)\le c\log{L}$, for some constant $c$ that will be chosen later.
In this section we prove that any such error $P$
will be corrected by the RG decoder for any  topological stabilizer code
satisfying the no-strings rule.
For completeness, we define the no-strings rule formally.

\subsection{No-strings rule}
\label{subs:NoStrings}

Let $P$ be any Pauli operator supported on a cube with linear size at most $\ltqo$,
and let $S=S(P)$ be its syndrome.
Suppose $S$ can be partitioned into a pair of disjoint
clusters, $S=C_1\cup C_2$, such that $d(C_a)\le r$ and $d(C_1,C_2)\ge \alpha r$
for some $r \ge 1$ and $\alpha \ge 1$.
We call such $P$ a {\em string segment} with {\em aspect ratio} $\alpha$.
Intuitively, a string segment with a large aspect ratio, $\alpha\gg 1$,
is an operator capable of creating a pair of well-separated clusters of defects from the vacuum.
We say $P$ is a {\em trivial} string segment if the clusters of defects $C_1,C_2$ are
neutral (note that the full syndrome $S=C_1\cup C_2$ is always neutral due to our assumptions on $P$).
Throughout this section we only consider topological stabilizer codes following Definition~\ref{dfn:TQO}.
Let $\alpha\ge 1$ be any constant.
\begin{defn}
A family of topological stabilizer codes obeys the {\em no-strings rule} with a constant $\alpha$
if any logical string segment with the aspect ratio greater than $\alpha$ is trivial.
\end{defn}
The 3D Cubic Code obeys no-string rule with $\alpha=5$ under the current $\ell_\infty$-metric convention~\cite{Haah2011Local}.

\subsection{Errors with a logarithmic energy barrier}
\label{subs:GoodErrors}

In order to assess the capability of the RG decoder rigorously,
we will need some of the results from Ref.~\cite{BravyiHaah2011Energy}.
Assume throughout this section that a family of topological stabilizer codes $\{ {\cal C}_L\}_L$
with $\ltqo \ge L^\gamma$ obey the no-strings rule with some constant $\alpha$.
Define \[ \xi(p)= (10\alpha)^p \] for notational convenience.
\begin{defn}[\cite{BravyiHaah2011Energy}]
A syndrome $S$ is called {\em sparse at level $p$} if $S$ can be partitioned
into disjoint clusters, $S=C_1\cup \ldots \cup C_n$, such that
\[
d(C_a)\le \xi(p) \quad \mbox{and} \quad  d(C_a\cup C_b)>\xi(p+1)
\]
for all $a\ne b$. Otherwise, $S$ is called {\em dense
 at level $p$}.
\end{defn}
(The term `dense' should not be confused with the `dense subset' of a topological space,
with which any nonempty open set intersects;
the two notions are even conceptually different.
Here, the denseness is just the negation of sparseness.)

\begin{prop}[Lemma~1 of \cite{BravyiHaah2011Energy}]
\label{prop:dense}
If $S$ is dense at all levels $q=0,\ldots,p$ then $S$ contains
at least $p+2$ defects.
\end{prop}

Let $\Gamma=(P_0,P_1,\ldots,P_t)$ be an error path implementing $P$
with the energy cost $m(\Gamma)=m$, see Section~\ref{subs:EnergyBarrier}.
Here $P_0=I$, $P_t=P$, while $E_j\equiv P_j P_{j-1}$ are single-qubit Pauli operators
for all $j$.
A sequence of syndromes $S(j)\equiv S(P_j)$, $j=0,\ldots,t$ is called a {\em level-$0$ syndrome history}.
Note that $S(0)=0$ and $S(t)=S$.
For any integer $p\ge 1$ define inductively a {\em level-$p$ syndrome history}
as a subsequence of $\{S(j)\}_{j=0,\ldots,t}$ including only $S(0)$, $S(t)$ and all syndromes $S(j)$,
$0<j<t$  which are dense at all levels $q=0,\ldots,p-1$.
Proposition~\ref{prop:dense} implies that
there is a level $p_{max}\le m-1$ such that
the level-$p_{max}$ syndrome history contains only the initial empty syndrome $S(0)=0$
and the final syndrome $S(t)=S$. The key technical result is the following.
\begin{prop}[Lemma~2 of \cite{BravyiHaah2011Energy}]
\label{prop:localization}
Let $S'$ and $S''$ be any consecutive pair of syndromes in the level-$p$ syndrome history.
Let $E$ be the product of all errors $E_j$ that occurred in the error path $\Gamma$ between $S'$ and $S''$.
If $4m (2+\xi(p)) < \ltqo$ then there exists a stabilizer $G\in {\cal G}$ such that
$E\cdot G$ is supported on the
$\xi(p)$-neighborhood of $S'\cup S''$.
\end{prop}

This is useful to characterize a low energy barrier error.
\begin{lem}
\label{lemma:neutral-components}
Let $P$ be any Pauli error,
$S=S(P)$ be its syndrome,
and $m=\Delta(P)$ be its energy barrier.
Suppose
\begin{equation}
\label{msmall}
8m(10\alpha)^m < \ltqo.
\end{equation}
Then, there exists a stabilizer $G \in {\cal G}$ such that $P \cdot G$
has support on the $\xi(m)$-neighborhood of $S$.
Any $R$-connected component of $S$ is neutral for $2\xi(m) < R \le 4 \xi(m)$.
\end{lem}
\begin{proof}
Let us apply Proposition~\ref{prop:localization} to the level $p_{max}$,
the smallest integer such that the syndrome history at that level
has only the initial and final syndrome.
Since $p_{max} < m$, the condition $4m(2+\xi(p_{max})) < \ltqo$ is satisfied.
We have $S'=0$, $S''=S$, and $E=P$.
Hence, there exists a stabilizer $G\in {\cal G}$ such that
$P\cdot G$ is supported on the $\xi(p_{max})$-neighborhood of $S$.
It proves the first statement of the lemma.

Let $r = \xi(m) = (10\alpha)^m$.
Choose any $R$ such that $2r < R \le 4r$
and let $C_a$ be any $R$-connected component of $S$.
Since $C_a$ contains at most $m$ defects, the diameter of $C_a$ is at most $mR$.
Restricting $P \cdot G$ on the $r$-neighborhood of $C_a$,
we obtain a Pauli operator $P_a$
supported on a cube of linear size at most $mR+2r \le 4rm+2r <\ltqo$ by assumption.
Furthermore, the support of $P_a$ is separated from $(P_a)^{-1} (P \cdot G)$
by distance at least $R - 2r > 0$.
Hence, $P_a$ creates the cluster $C_a$ from the vacuum.
Therefore, $C_a$ is neutral.
\end{proof}

We wish to have a well-separated cluster decomposition.
\begin{lem}
\label{lemma:wellseparated-decomposition}
Let $S$ be any cluster of $m>0$ defects.
For any integer $\mu \ge 1$, there exists a nonnegative integer $ p < m$ 
and a decomposition $S = C_1 \cup \cdots \cup C_n$ such that
$d(C_a) \le 4^p \mu$ and $d(C_a,C_b) > 2 \cdot 4^{p} \mu$ for all $a \neq b$.
\end{lem}
\begin{proof}
Let us say that a partition of $S$ into clusters is $\mu$-good if it satisfies the statement.
By grouping all defects occupying the same elementary cube into a cluster,
one obtains a partition $S=C_1\cup \ldots \cup C_g$. Obviously, $g\le m$, and $d(C_a) \le \mu$.
If this partition is not $0$-good,
then $g\ge 2$ and there is a pair, say, $C_1,C_2$ such that $d(C_1, C_2) \le 2\mu$.
Merging $C_1$ and $C_2$ into a single cluster $C'_2$,
one obtains a partition $S=C'_2 \cup C_3 \cup \ldots \cup C_g$ where $d(C'_2) \le 4\mu$.
If this partition is not $1$-good, then $g\ge 3$ and one can repeat the merging again.
After at most $g-1$ iterations, one arrives at a good partition.
\end{proof}
Note that the minimal enclosing boxes of distinct cluster do not overlap, since
\[
 d(b(C_a),b(C_b)) > 2 \cdot 4^p \mu - 4^p \mu - 4^p \mu = 0 .
\]
The following is the promised property of the RG decoder.
\begin{lem}
\label{lemma:RGdecoder-capability}
Let $P$ be any Pauli error with energy barrier $m=\Delta(P)$. Suppose
\[
 (80\alpha)^m < \ltqo .
\]
Then calling the RG decoder on the syndrome $S(P)$ returns a correcting operator
$P_{ec}$ such that $P P_{ec}$ is a stabilizer.
\end{lem}
Thus, the RG decoder corrects $P$ if $\Delta(P) < \frac{\gamma}{\log (80 \alpha)} \log L$.
\begin{proof}
Let $S=S(P)$ be the syndrome.
Let $p$ be the integer such that $2 \xi(m) < 2^p \le 4 \xi(m)$.
Setting $\mu=2^p$ in Lemma~\ref{lemma:wellseparated-decomposition},
$S$ is decomposed into $S=C_1 \cup \ldots \cup C_n$
such that $d(C_a) \le 2^{p'}$ and $ d(C_a,C_b) > 2^{p'+1}$ for all $a \ne b$,
where $p'$ is an integer such that $p \le p' < 2m+p$.
By Lemma~\ref{lemma:neutral-components},
each $C_a$ is neutral for being a disjoint union of neutral $2^p$-connected components.
The RG subroutines EC($s$) with $s=0,1,\ldots,p-1$,
can only annihilate some neutral $2^s$-connected components of $C_a$,
which does not alter the neutrality of $C_a$.
Therefore, the RG decoder from level-$0$ to $p$
will annihilate each cluster $C_a$, and hence $S$ at last.

We need to show that $P \cdot P_{ec}$ is a stabilizer,
where $P_{ec}$ is the returned correcting operator.
Let $B_a$ be the $(10\alpha)^m$-neighborhood of $b(C_a)$.
Our assumptions imply that $B_a$ has diameter smaller than $\ltqo$
and distinct $B_a$'s do not intersect.
By construction, the operators $P_{ec}$ and $P \cdot G$
have support in the union $B_1 \cup \cdots \cup B_n$.
Therefore, $P \cdot G \cdot P_{ec} = Q_1 \cdots Q_n$,
where $Q_a$ has support on $B_a$ and has trivial syndrome.
Topological order condition implies that $Q_a$ are stabilizers, so is the product.
\end{proof}

As was mentioned earlier, the full hierarchy of the RG decoder
is not necessary to correct the error with the low energy barrier.
A single level-$p$ error correction with $p$ proportional to $\log \ltqo$, will be sufficient.
We nevertheless include the hierarchy
since in practice it corrects errors with slightly higher
(although only by a constant factor) energy barrier
at a marginal slowdown of the decoder.
If one wishes to apply the decoder against random errors,
the hierarchy becomes necessary.
We discuss the character of the RG decoder against random errors in Appendix.

\subsection{Further specialization}

A closer analysis reveals a simplification of TestNeutral defined in Section~\ref{subs:rgdecoder} for the 3D Cubic Code.
We defined TestNeutral to return the identity operator if a cluster turns out to be charged.
The modified TestNeutral$'$ just applies the broom algorithm and returns recorded operator in any case.
It gives the {\em same} characteristic as stated in the Lemma~\ref{lemma:RGdecoder-capability}.
EC$'(p)$ using TestNeutral$'$ will transform a charged cluster $C_a$ to a {\em different} cluster $C'_a$,
but $C'_a$ is still contained in $b(C_a)$.
Due to Lemma~\ref{lemma:wellseparated-decomposition}, $b(C_a)$ do not overlap
at a high level $p$, and EC$'(p)$ will eliminate neutral clusters at last.
This specialized version of RG decoder is used in our numerical simulation of the next section.

\section{Numerical Simulation}
\label{sec:numerics}

We test our theoretical bound by numerical simulations. The interaction of the memory system with a thermal bath is simulated by Metropolis evolution. As we wish to observe low temperature behavior we adopt continuous time algorithm by Bortz, Kalos, and Lebowitz (BKL)~\cite{BortzKalosLebowitz1975}.
A pseudo-random number generation package \verb|RngStream| by L'Ecuyer~\cite{LEcuyerEtAl2002RngStream} was used.
As before, the coupling constant in the Hamiltonian is set to $J = 1/2$ so a single defect has energy 1.
Although the 3D Cubic Code is inherently quantum,
it is relevant to consider only $X$-type errors (bit flip) in the simulation,
thanks to the duality of the $X$- and $Z$-type stabilizer generators of the 3D Cubic Code.
The simulation thus is purely classical.
The errors are represented by a binary array of length $2L^3$,
and the corresponding syndrome by a binary array of length $L^3$.

The memory time is measured to be the first time when the memory becomes unreliable. There are two cases the memory is unreliable: either the broom algorithm fails to remove all the defects so we have to reinitialize the memory, or a nontrivial logical error is occurred. It is thus necessary in our simulation to keep track of the error operator during the time evolution.
In fact, most of the time, it was the broom algorithm's failure that made the memory unreliable.
Nontrivial logical errors occurred only for very small system sizes $L=5,7$.

It is too costly to decode the system every time it is updated. 
Alternatively, we have performed a trial decoding  every fixed time interval
\[
 T_{ec} = \frac{e^{4\beta}}{100}
\]
where $\beta$ is the inverse temperature.
Although the time evolution of the BKL algorithm is stochastic, a single BKL update typically advances time much smaller than $T_{ec}$. So it makes sense to decode the system every $T_{ec}$. The exponential factor appears naturally because BKL algorithm advances time exponentially faster as $\beta$ increases. It is to be emphasized that we do not alter the system by the
trial decodings (a copy of the actual syndrome has been created for each trial decoding). 

The system sizes $L^3$ for the simulation are chosen such that the code space dimension is exactly $2$, for which the complete list of logical operators is known. If the linear size $L$ is $\le 200$, this is the case when $L$ is not a multiple of 2, 15, or 63 \cite{Haah2011Local}. For these system sizes, to check whether a logical operator is nontrivial is to compute the commutation relation with the known nontrivial logical operators.

The measured memory time for a given $L$ and $\beta$ is observed to follow an exponential distribution; a memory system is corrupted with a certain probability given time interval.
Specifically, the probability that the measured memory time is $t$ is proportional to $e^{-t/\tau}$.
Thus the memory time should be presented as the characteristic time of the exponential distribution. We choose the estimator for the characteristic time to be the sample average $\bar T = \frac{1}{n}\sum_i^n T_i$. The deviation of the estimator will follow a normal distribution for large number $n$ of samples. We calculated the confidence interval to be the standard deviation of the samples divided by $\sqrt{n}$.
For each $L$, $400$ samples when $\beta \le 5.0$ and $100$ samples when $\beta > 5.0$ were simulated.
The result is summarized in Fig.~\ref{fig:TvsL},\ref{fig:TvsBeta}.

Figure~\ref{fig:TvsBeta} clearly supports $\log T_{mem} = \frac{1}{4}c \beta^2 + \cdots$. Figure~\ref{fig:TvsL} demonstrates the power law for small system size:
\[
 T_{mem} \propto L^{2.93 \beta -10.5}
\]
We wish to relate some details of the model with the numerical coefficients.
The Rigorous analysis of the previous section, gives a relatively small coefficient $c$ of the energy barrier for correctable errors by our RG decoder. However, we expect that the coefficient of $\beta$ in the exponent is the same as the constant $c$ that appear in the energy barrier
\[
 E = c \log_2 R
\]
to create an isolated defect separated from the other by a distance $R$.
This is based on an intuition that the output $P'$ of the decoder would have roughly the same support as the real error $P$ for the most of the time, provided that the error has energy barrier less than $\Delta=c \log_2 \ltqo$.
Thus, an error of energy barrier less than $\Delta$ would be corrected by the decoder.
Our empirical formula supports this intuition.
It suggests that $c = 2.93 \log 2 = 2.03 \sim 2$.

Indeed, we can illustrate explicitly an error path that separates a single defect from the rest by distance $2^p$ during which only $2p+4$ defects are needed. This is an improvement in the upper bound on the energy barrier for separating a charged set of defects, for which it was $c \le 4$ in \cite{BravyiHaah2011Energy}.

\begin{figure}
\centering
\includegraphics[width=.5\textwidth]{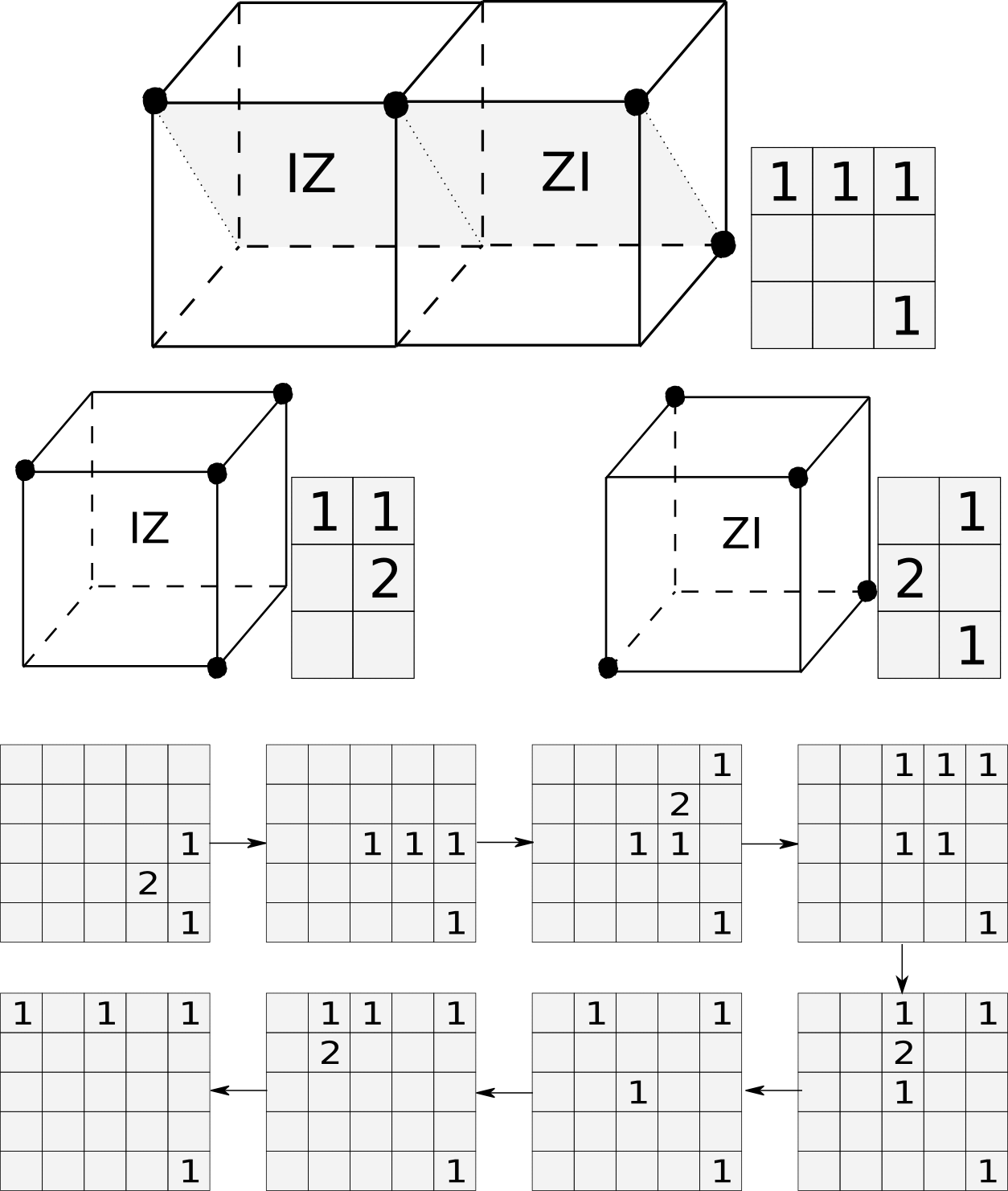}
\caption{Construction of a hook of level 2 from the vacuum. The grid diagram represents the position and the number of defects in the ($x=z$)-plane. For each transition, an operator of weight 1 is applied. The total number of defects never exceeds 6. From a level-$0$ hook (the second diagram in the sequence), a level-$1$ hook (the last in the sequence) is constructed using extra 2 defects.}
\label{fig:hook}
\end{figure}

Consider an error of weight 2 that creates 4 defects as shown in the top of Fig.~\ref{fig:hook}. We call it the {\em level-0 hook}. The bottom sequence depicts a process to create a configuration shown at the bottom-left, which we call {\em level-1 hook}. One sees that level-$1$ hook is similar with ratio 2 to level-$0$, and is obtained from level-0 with extra 2 defects.
One defines level-$p$ hooks hierarchically.
We claim that a level-$p$ hook can be constructed from the vacuum using $2p+4$ defects. The proof is by induction. The case $p=1$ is treated in the diagrams. Suppose we can construct level-$p$ hook using $2p+4$ defects. Consider the $2^{nd}$, $4^{th}$, $6^{th}$, and $8^{th}$ steps in Fig.~\ref{fig:hook}. They can be viewed as a minuscule version of level-$p$ steps that construct a level-$(p+1)$ hook from the level-$p$ hooks. It requires at most $2p+4+2$ defects to perform the level-$p$ step; this completes the induction.
It may not be obvious whether a high level hook corresponds to a nontrivial logical operator,
but such a large hook is bad enough to make our decoder fail.  

\vspace{1cm}
\centerline{\bf Acknowledgments}
We would like to thank David DiVincenzo, John Preskill, and Barbara Terhal for helpful discussions.
SB is supported in part by the DARPA QuEST program under contract number HR0011-09-C-0047, and
IARPA QCS program under contract number D11PC20167.
JH is supported in part by the Korea Foundation for Advanced Studies and by the Institute for Quantum Information and Matter, an NSF Physics Frontiers Center.
Computational resources for this work were provided by IBM Blue Gene Watson supercomputer center.

\appendix

\section{Benchmark of the decoder}
\label{sec:benchmark}

\begin{figure}
\centering
\begin{minipage}{.45\textwidth}
\includegraphics[width=\textwidth]{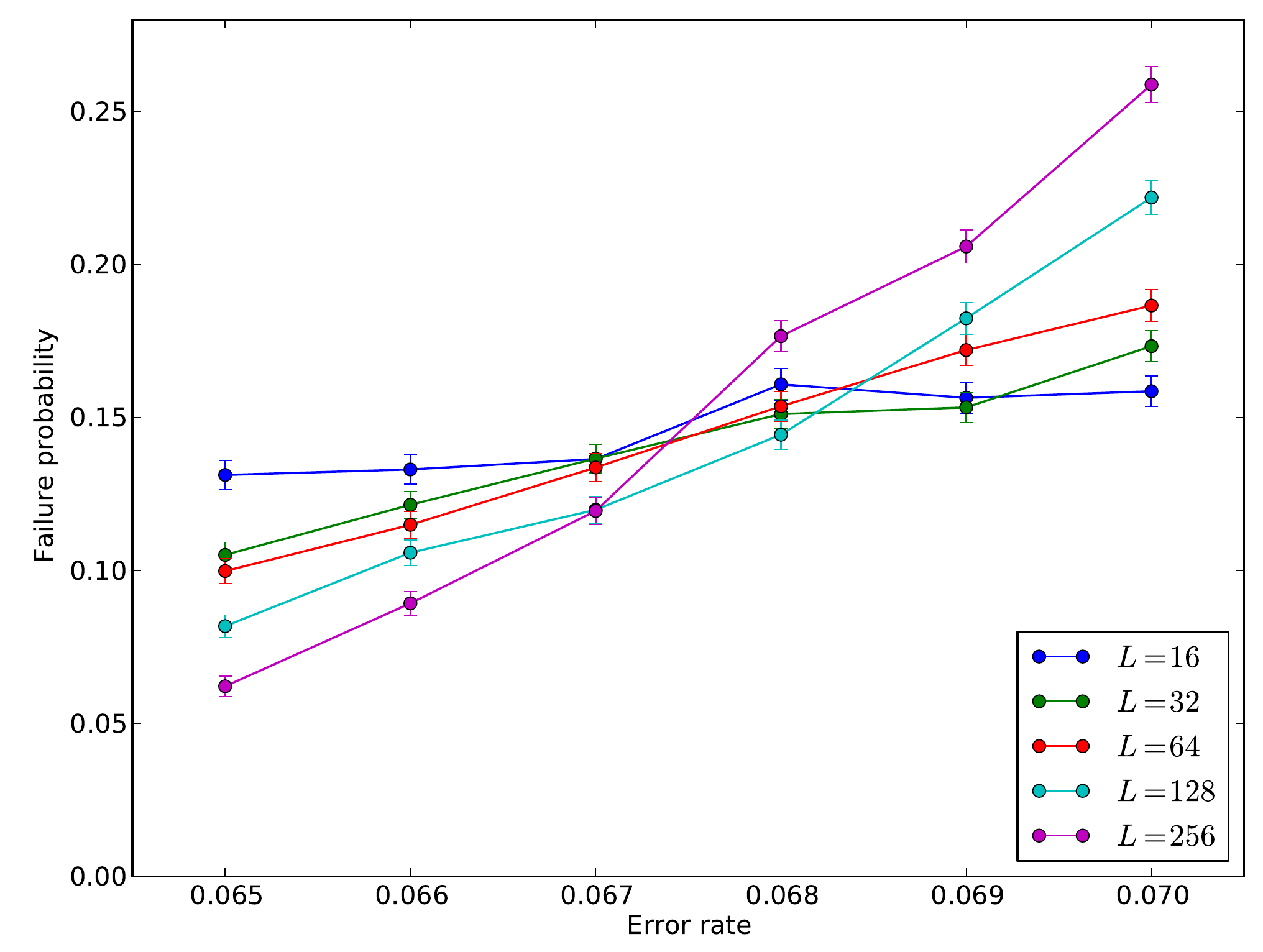}
\end{minipage}
\begin{minipage}{.45\textwidth}
\includegraphics[width=\textwidth]{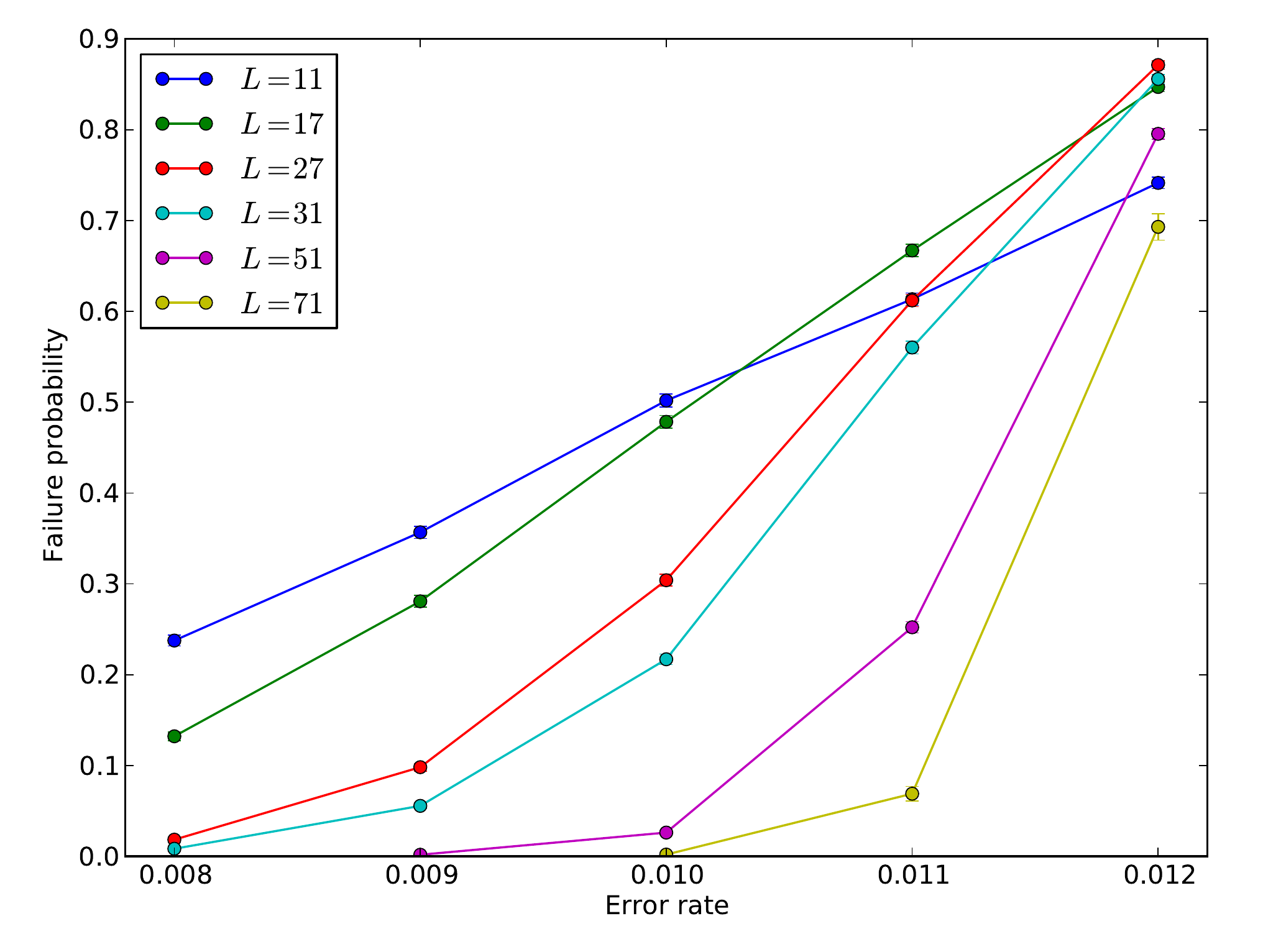}
  \end{minipage}
\caption{
The thresholds of 2D toric code (left) and 3D cubic code (right)
under independent random bit-flip errors using our RG decoder,
are measured to be $p_c(\text{2D toric})=6.7(1)\%$ and $p_c(\text{3D cubic}) \gtrsim 1.1 \%$.
}
\label{fig:threshold}
\end{figure}

Given a decoder, a family of quantum codes indexed by code length (system size)
is said to have an \emph{error threshold} $p_c$
if the probability for decoder to fail approaches zero in the limit of large code length
provided the random error rate $p$ is less than $p_c$.
We tested our decoder with respect to random uncorrelated bit-flip errors on the well-studied 2D toric code.
The error threshold is measured to be $6.7(1)\%$. See Fig.~\ref{fig:threshold}.
It is roughly two-thirds of the best known value $10.3\%$
based on the perfect matching algorithm~\cite{DennisKitaevLandahlEtAl2002Topological, Harrington2004thesis},
or $9\%$ based on a renormalization group decoder of similar nature to ours~\cite{Duclos-CianciPoulin2009Fast}.
This is remarkable for our decoder's simplicity and applicability.
The 3D cubic code has threshold $\gtrsim 1.1\%$ under independent bit-flip errors.

\section{Threshold theorem for topological stabilizer codes}
\label{sec:threshold}

In this section we prove that any topological stabilizer code can tolerate stochastic local errors with
a small constant rate assuming that the error correction is performed using the RG decoder.
We assume without lose of generality that each stabilizer generator is supported on a unit cube. Each site of the lattice may contain finitely many qubits. A generator at a cube $c$ may act only on qubits of $c$.
We shall assume that errors at different sites are independent and identically distributed.
More precisely, let $E(P)$ be the set of sites at which a Pauli error $P$ acts non-trivially. We shall assume that
\begin{equation}
 \mathrm{Pr}[E(P)=E] = (1- \epsilon)^{V-|E|} \epsilon^{|E|}
\label{eq:random-error-distribution}
\end{equation}
where $0 \le \epsilon \le 1$ is the error rate and $V=L^D$ is the total number of sites (the volume of the lattice).
For example, the depolarizing noise in which every qubit experiences $X,Y,Z$ errors
with the probability $p/3$ each, satisfies
Eq.~(\ref{eq:random-error-distribution}) with the error rate $\epsilon=1-(1-p)^q$, where $q$ is the number of qubits per site.
\begin{theorem}
Suppose a family of stabilizer codes has topological order, see Definition~\ref{dfn:TQO}.
Then, there exists a constant threshold $\epsilon_0>0$
such that for any $\epsilon<\epsilon_0$ the RG decoder corrects random independent errors with rate $\epsilon$
with the failure probability at most $e^{-\Omega(L^\eta)}$ for some constant $\eta>0$.
\label{thm:universal-threshold}
\end{theorem}
In the rest of this section we prove the theorem.
 Our proof borrows some techniques from \cite{Gray2001Guide,Gacs1986Reliable,Harrington2004thesis}, specifically Section~5.1 of \cite{Gray2001Guide}.  For reader's convenience, we briefly  summarize the RG decoding algorithm below (see
Section~\ref{sec:rgdecoder} for details).

Recall that we use $\ell_\infty$-metric. A cube of linear size $r$ thus has diameter $r$.
Let $S=S(P)$ be the syndrome of a Pauli error $P$ considered as a set of defects (violated stabilizers).
A subset $M\subseteq S$ is called $r$-connected iff $M$ cannot be partitioned
into a pair of disjoint proper subsets separated by distance more than $r$. A maximum $r$-connected
subset of $S$ is called an $r$-connected component of $S$. A cluster of defects $M\subseteq S$
is called  \emph{neutral} if it can be created from the vacuum by a Pauli operator $P$ supported in a cube of linear size $\ltqo$.
The smallest rectangular box enclosing a cluster $M$ will be denoted $b(M)$.
Any neutral cluster $M$ can be created by a Pauli operator supported on the $1$-neighborhood of $b(M)$,
see Definition~\ref{dfn:TQO}.

The \emph{level-$p$ error correction} {\em EC($p$)} on a syndrome $S$ is the following subroutine.
(i) find all neutral $2^p$-connected components $M$ of $S$,
(ii) for each $M$ found at step 1, calculate and apply a Pauli operator $P$ supported on the 1-neighborhood of $b(M)$ that annihilates $M$, and update the syndrome accordingly.
Calling the full RG decoder on a syndrome $S$ involves the following steps:
(i) run EC(0), EC(1), ..., EC($\lfloor \log_2 \ltqo \rfloor$),
(ii) if the resulting syndrome $S$ is empty, return the accumulated Pauli operator applied by the subroutines EC($p$). Otherwise, declare a failure.

Below we shall use the term `error'
both for the error operator $P$ and for the subset of sites $E$ acted on by $P$, whenever the meaning is clear from the context.
Let us choose an integer $Q\gg 1$ and define a class of errors which are properly corrected by the RG decoder,
see Lemma~\ref{lem:correctability-criterion} below.
We will see later that this class of errors actually includes all errors which are likely to appear for small enough error rate.
\begin{defn}
Let $E$ be a fixed error.
A site $u\in E$ is called a \emph{level-$0$ chunk}. A non-empty subset of $E$ is called a \emph{level-$n$ chunk} ($n\ge 1$) if it is a disjoint union of two level-$(n-1)$ chunks and its diameter is at most $Q^n /2$.
\end{defn}
The term `chunk' is chosen in order to avoid confusion with `cluster', which is used for a set of defects.
Note that a level-$n$ chunk contains exactly  $2^n$ sites.
Given an error $E$, let $E_n$ be the union of all level-$n$ chunks of $E$.
If $u \in E_{n+1}$, then by definition $u$ is an element of a level-$(n+1)$ chunk.
Since a level-$(n+1)$ chunk is a union of two level-$n$ chunks, $u$ is contained in a level-$n$ chunk. Hence, $u \in E_n$, and the sequence $E_n$ form a descending chain
\[
 E = E_0 \supseteq E_1 \supseteq \cdots \supseteq E_m ,
\]
where $m$ is the smallest integer such that $E_{m+1}=\emptyset$.
Let $F_i = E_i \setminus E_{i+1}$, so $E = F_0 \cup F_1 \cup \cdots \cup F_m$ is expressed as a disjoint union, which we call the \emph{chunk decomposition} of $E$.
\begin{prop}
Let $Q \ge 6$ and $M$ be any $Q^n$-connected component of $F_n$. Then
$M$ has diameter $\le Q^n$ and is separated from $E_n \setminus M$ by distance $> \frac{1}{3}Q^{n+1}$.
\label{prop:structure-Fn}
\end{prop}
\begin{proof}
We claim that for any pair of sites $u \in F_n = E_n \setminus E_{n+1}$ and $v \in E_n$ we have $ d( u,v) \le Q^n$ or $d( u,v ) > \frac{1}{3} Q^{n+1}$.
Suppose on the contrary to the claim, that there is a pair $u \in F_n$ and $v \in E_n$ such that $ Q^n < d(u,v) \le Q^{n+1} /3$. Let $C_u \ni u$ and $C_v \ni v$ be level-$n$ chunks that contains $u$ and $v$, respectively. Since the diameters of $C_{u,v}$ are $\le Q^n /2$ and $d(u,v) > Q^n$, we deduce that $C_u$ and $C_v$ are disjoint. On the other hand,
\[
 d( C_u \cup C_v ) \le d( u,v ) + d(C_u) + d(C_v)  \le Q^{n+1}/2
\]
since $Q \ge 6$. Thus, $C_u \cup C_v$ is a level-$(n+1)$ chunk that contains $u$
which shows that $u \in E_{n+1}$. It contradicts to our assumption that $u\in F_n=E_n \setminus E_{n+1}$.
\end{proof}
Note that in the chunk decomposition a $Q^n$-connected component $P$ of $E_n$ may not be separated from the rest $E \setminus P$ by distance $> Q^n$.
\begin{lem}
Let $Q \ge 10$. If the length $m$ of the chunk decomposition of an error $E$ satisfies $Q^{m+1} < \ltqo$, then $E$ is corrected by the RG decoder.
\label{lem:correctability-criterion}
\end{lem}
\begin{proof}
Consider any fixed error $P$ supported on a set of sites $E$.
Let $E=F_0 \cup F_1 \cup \cdots \cup F_m$ be the chunk decomposition of $E$, and
let $F_{j,\alpha}$ be the $Q^j$-connected components of $F_j$.
Also, let $B_{j,\alpha}$ be the $1$-neighborhood of the smallest box enclosing the syndrome
created by the restriction of $P$ onto $F_{j,\alpha}$.
Proposition~\ref{prop:structure-Fn} implies that
\begin{equation}
\label{subchunks2}
d(B_{j,\alpha})\le Q^j+2 \quad \mbox{and} \quad d(B_{j,\alpha}, B_{k,\beta})>\frac13 Q^{1+\min{(j,k)}}-2.
\end{equation}
Let $P_{ec}^{(p)}$ be the accumulated correcting operator returned by the levels $0,\ldots,p$  of the RG decoder.
Let us use induction in $p$ to prove the following statement.
\begin{enumerate}
\item The operator $P_{ec}^{(p)}$ has support on the union of the boxes $B_{j,\alpha}$.
\item  The operators $P_{ec}^{(p)}$ and $P$ have the same restriction on $B_{j,\alpha}$ modulo stabilizers
for any $j$ such that $2^p\ge Q^j+2$.
\end{enumerate}
The base of induction is $p=0$.  Using Eq.~(\ref{subchunks2}) we conclude that any
$1$-connected component of the syndrome $S(P)$ is fully contained
inside some box $B_{j,\alpha}$. It proves that $P_{ec}^{(0)}$ has support on the union of the boxes $B_{j,\alpha}$.
The second statement is trivial for $p=0$.

Suppose we have proved the above statement for some $p$.
Then the operator $P\cdot P_{ec}^{(p)}$ has support only inside
boxes $B_{j,\alpha}$ such that $2^p<Q^j+1$ (modulo stabilizers). It follows that any
$2^{p+1}$-connected component of the syndrome caused by $P\cdot P_{ec}^{(p)}$ is contained in some box
$B_{j,\alpha}$ with $2^p<Q^j+1$. Note that the RG decoder never adds new defects;
we just need to check that $2^{p+1}$-connected components do not cross the boundaries
between the boxes $B_{j,\alpha}$ with $2^p<Q^j+1$. This follows from Eq.~(\ref{subchunks2}).
Hence $P_{ec}^{(p+1)}$
has support in the union of $B_{j,\alpha}$. Furthermore, if
$2^p<Q^j+1\le 2^{p+1}$, the cluster of defects created by $P\cdot P_{ec}^{(p)}$
inside $B_{j,\alpha}$ forms a single  $2^{p+1}$-connected component
of the syndrome examined by  EC$(p+1)$.
This cluster is neutral since we assumed $Q^{m+1}<L_{tqo}$.
Hence $P_{ec}^{(p+1)}$ will annihilate this cluster.
The annihilation operator is equivalent to the restriction of $P\cdot P_{ec}^{(p)}$
onto $B_{j,\alpha}$ modulo stabilizers, since the linear size
of $B_{j,\alpha}$ is smaller than $\ltqo$. It proves the induction hypothesis
for the level $p+1$. 
\end{proof}

The preceding lemma says that errors by which the RG decoder could be confused are those from very high level chunks. What is the probability of the occurrence of such a high level chunk if the error is random according to Eq.\eqref{eq:random-error-distribution}? Since our probability distribution of errors depend only on the number of sites in $E$, this question is completely percolation-theoretic.

Let us review some terminology from the percolation theory\cite{Grimmett1999Percolation}. An event is a collection of configurations. In our setting, a configuration is a subset of the lattice. Hence, we have a partial order in the configuration space by the set-theoretic inclusion. An event $\mathcal E$ is said to be \emph{increasing} if $E \in \mathcal{E}, E \subseteq E'$ implies $E' \in \mathcal{E}$. For example, the event defined by the criterion that there exists an error at $(0,0)$, is increasing.
The \emph{disjoint occurrence} $\mathcal A \circ \mathcal B$ of the events $\mathcal A$ and $\mathcal B$ is defined as the collection of configurations $E$ such that $E = E_a \cup E_b$ is a disjoint union of $E_a \in \mathcal A$ and $E_b \in \mathcal B$. To illustrate the distinction between $\mathcal A \circ \mathcal B$ and $\mathcal A \cap \mathcal B$, consider two events defined as $\mathcal A = $ ``there are errors at $(0,0)$ and at $(1,0)$'', and $\mathcal B = $ ``there are errors at $(0,0)$ and at $(0,1)$''. The intersection $\mathcal A \cap \mathcal B$ contains a configuration $\{ (0,0),(1,0),(0,1) \}$, but the disjoint occurrence $\mathcal A \circ \mathcal B$ does not. A useful inequality by van den Berg and Kesten (BK) reads \cite{BergKesten1985, Grimmett1999Percolation}
\begin{equation}
 \mathrm{Pr}[ \mathcal A \circ \mathcal B ] \le \mathrm{Pr}[ \mathcal A ] \cdot \mathrm{Pr}[ \mathcal B]
\label{eq:BK-inequality}
\end{equation}
provided the events $\mathcal A$ and $\mathcal B$ are increasing.

\begin{proof}(of Theorem~\ref{thm:universal-threshold})
Consider a $D$-dimensional lattice
and a random error $E$ defined by Eq.~(\ref{eq:random-error-distribution}).
Let $B_n$ be a fixed cubic box of linear size $Q^n$ and
$B_n^+$ be the  box of linear
size $3Q^n$ centered at $B_n$.
Define the following probabilities:
\begin{eqnarray}
p_n &=& \mathrm{Pr}\left[ \mbox{$B_n$ has a non-zero overlap with a level-$n$ chunk of $E$}\right] \nn \\
\tilde{p}_n &=& \mathrm{Pr}\left[\mbox{$B_n^+$ contains a level-$n$ chunk of $E$}\right] \nn \\
q_n &=& \mathrm{Pr}\left[\mbox{$B_n^+$ contains $2$ disjoint  level-$(n-1)$ chunks of $E$}\right] \nn \\
r_n &=& \mathrm{Pr}\left[\mbox{$B_n^+$  contains a level-$(n-1)$ chunk of $E$}\right] \nn
\end{eqnarray}
Note that all these probabilities do not depend on the choice of the box $B_n$
due to translation-invariance. Since a level-$0$ chunk is just a single site
of $E$, we have $p_0=\epsilon$.
We begin by noting that
\[
p_n\le \tilde{p}_n\le q_n.
\]
Here we used  the fact that
any level-$n$ chunk has diameter at most $Q^n/2$ and that any level-$n$ chunk
consists of a disjoint pair of level-$(n-1)$ chunks.
Let us fix the box $B_n^+$ and let ${\cal Q}_n$ be the event that
$B_n^+$ contains a disjoint pair of  level-$(n-1)$ chunks of $E$.
Let ${\cal R}_n$ be the event that $B_n^+$ contains a level-$(n-1)$ chunk of $E$.
Then ${\cal Q}_n={\cal R}_n\circ {\cal R}_n$.
It is clear that  ${\cal Q}_n$ and ${\cal R}_n$  are increasing events.
Applying the van den Berg and Kesten inequality we arrive at
\[
q_n\le r_n^2.
\]
Finally, since $B_n^+$ is a disjoint union of $(3Q)^D$ boxes of linear size $Q^{n-1}$,
the union bound yields
\[
r_n\le (3Q)^D p_{n-1}.
\]
Combining the above inequalities we get
$p_n \le (3Q)^{2D} p_{n-1}^2$, and hence
\[
p_n \le (3Q)^{-2D}((3Q)^{2D} \epsilon)^{2^n}.
\]
The probability $p_n$ is  doubly exponentially small in $n$
whenever $\epsilon < (3Q)^{-2D}$.
 If there exists at least one level-$n$ chunk, there is always a box of linear size $Q^n$ that overlaps with it. Hence, on the finite system of linear size $L$, the probability of the occurrence of a level-$m$ chunk is bounded above by $L^D p_m$.
 Employing
 Lemma~\ref{lem:correctability-criterion}, we conclude that the RG decoder fails with probability at most $p_{fail}=L^D p_m$
  for any $m$ such that $Q^{m+1} < \ltqo$. Since we assumed that $\ltqo\ge L^\gamma$,
 one can choose $m\approx \gamma \log{L}/\log{Q}$.
 In this case $p_{fail}=\exp{(-\Omega(L^\eta))}$ for $\eta \approx \gamma/\log{Q}$.
 We have proved our theorem with $\epsilon_0 = (3Q)^{-2D}$ where $Q = 10$.
\end{proof}

\end{document}